\documentclass[runningheads]{llncs}
\usepackage{graphicx}

\usepackage[ruled,linesnumbered]{algorithm2e}

\usepackage[utf8]{inputenc}
\usepackage{amsmath} 
\usepackage{amsfonts} 
\usepackage{amssymb}
\usepackage{enumerate}
\usepackage{mathtools}
\usepackage{xcolor}
\usepackage{cleveref}
\usepackage{subcaption}
\usepackage{custom_commands}
\usepackage{booktabs}
\usepackage{float}

\usepackage{enumerate}

\DeclarePairedDelimiter\abs{\lvert}{\rvert}%

\begin{document}
\title{The Role of Heterogeneity in Autonomous Perimeter Defense Problems\thanks{Supported by the Army Research Laboratory as part of the Distributed and Collaborative Intelligent Systems and Technology (DCIST) Collaborative Research Alliance (CRA).}}
%
%
\author{ { Aviv Adler\inst{1}\orcidID{0000-0003-3698-7639} \and
Oscar Mickelin \inst{2}\orcidID{0000-0003-0167-1992} \and Ragesh K. Ramachandran \inst{3}\orcidID{0000-0003-4396-0127} \and
Gaurav S. Sukhatme\inst{3,5}\orcidID{0000-0003-2408-474\text{X}} \and Sertac Karaman \inst{3}\orcidID{0000-0002-2225-7275} }}
\authorrunning{A. Adler et al.}
%
\institute{Department of Electrical Engineering and Computer Science, Massachusetts Institute of Technology, Cambridge, MA, USA \email{adlera@mit.edu} \and
Program in Applied and Computational Mathematics, Princeton University, Princeton NJ 08544, USA
\email{hm6655@princeton.edu} \and
Department of Computer Science, University of Southern California, Los Angeles, CA 90089, USA.
\email{\{rageshku,gaurav\}@usc.edu} \and Department of Aeronautics and Astronautics, Massachusetts Institute of Technology, Cambridge, MA, USA, 02139 \email{sertac@mit.edu}\and G.S. Sukhatme holds concurrent appointments as a Professor at USC and as an Amazon Scholar. This paper describes work performed at USC and is not associated with Amazon.}
\maketitle              
\begin{abstract}
When is heterogeneity in the composition of an autonomous robotic team beneficial and when is it detrimental? We investigate and answer this question in the context of a minimally viable model that examines the role of heterogeneous speeds in perimeter defense problems, where defenders share a total allocated speed budget. We consider two distinct problem settings and develop strategies based on dynamic programming and on local interaction rules. We present a theoretical analysis of both approaches and our results are extensively validated using simulations. Interestingly, our results demonstrate that the viability of heterogeneous teams depends on the amount of information available to the defenders. Moreover, our results suggest a universality property: across a wide range of problem parameters the optimal ratio of the speeds of the defenders
remains nearly constant.

\keywords{Perimeter defense  \and Heterogeneous multi-robot team \and Dynamic Programming.}
\end{abstract}
\section{Introduction}
An increasingly 
important task, where a robotic system can be employed, is in defending an area against external agents, which pose varying levels of threat. Examples include defending airports against intruding and flight-grounding drones \cite{lykou2020defending}, defending wildlife habitats against trespassing poachers \cite{casey2014drones}, extinguishing and preventing the spread of devastating wildfires caused by human or natural activity \cite{fire}, as well as military applications \cite{military}. 

In general, solutions to perimeter defense problems allude to finding strategies for a set of agents restricted to the perimeter of an area, entrusted with defending the area from intruders which are trying to breach the perimeter of the area \cite{Shishika_2020_review}. 

Compared to a homogeneous team of robots, a team of robots with varying capabilities (heterogeneous team) comes with its unique set of advantages and challenges. Equipping different agents with different capabilities can lead to synergy effects where the heterogeneous system outperforms the alternative homogeneous system composed of identical agents. As a result, in the last decade, there has been significant interest in the robotics community to define, explore, and quantify heterogeneity in different robot applications \cite{Twu2014,Maria2018,Ramachandran2022,Sid2022RAL,Harish2020,ramachandran2019}.




This paper investigates the impact of heterogeneity in multi-robot teams for the perimeter defense problem. We propose two optimal strategies, valid under different assumptions. The first strategy is based on dynamic programming (DP) \cite{cormen01introduction}. It is optimal when the defenders are able to predict the location of the incoming attacks, but suffers from the curse of dimensionality and therefore relatively high associated computational costs. The second strategy is based on local interaction rules, and is optimal when the defenders have no information about the incoming attacks. This strategy can be efficiently computed in an online fashion, but does not implement any prior knowledge of the attack locations.

We prove the optimality of both strategies and analyze their time complexities. The algorithms are extensively validated on simulations. Our numerical experiments are two-dimensional, but the majority of the theoretical results remain valid for any dimension. This includes three-dimensional perimeters in applications involving drones, and higher-dimensional perimeters arising as constraint sets in a state space of arbitrary dimension.

Our results show that heterogeneity is beneficial in the case where the defenders have access to information about the incoming attacks, and is detrimental when the defenders have no information about the attacks. Moreover, we show the universality property that the optimal ratio of the speeds of the defenders remains nearly constant for a two defender case setting. 







\textbf{Related work:} Perimeter defense problems are a variant of pursuit-evasion problems which have been studied extensively in literature. The seminal work of Issacs delineates differential-game approaches to arrive at equilibrium strategies for one pursuer one evader games \cite{isaacs1965differential}. There has been considerable effort by researchers from various communities for solving various variants of pursuit-evasion games involving multiple pursuers and evaders \cite{Moll2020,Rui2019,Fuchs2010}. These papers contain works that view pursuit-evasion games either from the pursuers' side, from the evaders' side, or both. The curse of dimensionality poses a considerable challenge in solving problems involving multiple pursuers and evaders. The perimeter defense problem presented in this paper is a variant of the \textit{target guarding problem} first introduced by Isaacs~\cite{isaacs1965differential}. In the target guarding problem setting an agent is tasked with guarding an region of interest against an adversarial agent. Investigations on perimeter defense problems are in their nascent stage. The review paper by Shishika and Kumar \cite{Shishika_2020_review} delineates the recent works done on multi-robot perimeter defense problems \cite{Shishika2018,Lee2020,Shishika2020RAL,Shishika2019}. Unlike the problems considered in these works, we consider a class of perimeter defense problems in which the number of attackers is much larger than the number of defenders. 





The remainder of the paper is organized as follows. \Cref{sec::problem_statement} contains our notation together with the problem statement. \Cref{sec::theoretical_results,sec::unit-horizon} detail our theoretical results in the infinite and unit-time horizon cases respectively. \Cref{sec::simulation_results} concludes with simulation results.

\begin{table}[!t]
	\centering
	\caption{Notations}
	\label{tab:notation}
	\begin{tabular}{@{}ll@{}}
		\toprule
		Symbol         & Description                    \\ \midrule
		$\cX$  & Perimeter               \\
		$m$          & Number of defenders               \\
		$n$          & Number of attacks               \\
		$x_i$ & Location of defender $i$ \\
		$v_i$ & Speed of defender $i$, ordered decreasingly \\
		$z_j$ & Location of attack $j$ \\
		$t_j$ & Time of attack $j$, ordered increasingly \\
        $h$   & Defender horizon \\
        $\opt(\bv, \{(z_j, t_j)\}_{j=1}^n)$ & Minimum number of attacks the defenders can let through \\
		\bottomrule
	\end{tabular}
	\vspace{-0.3cm}
\end{table}





\section{Problem statement}\label{sec::problem_statement}
In this paper, bold letters are used to represent vectors and non-bold letters to represent scalars. Calligraphic letters are used to represent sets, and $\abs{\mathcal{S}}$ denotes the cardinality of a set $\mathcal{S}$.

For any positive integer $n \in \mathbb{Z}^+$, $[n]$ denotes the set $\{1,2, \cdots, n\}$. For a domain $\cX$ with $x_1, x_2 \in \cX$, $\text{dist}(x_1,x_2)$ denotes the length of the shortest path between $x_1$ and $x_2$ contained inside $\cX$. As an example, in the case when $\cX$ denotes a circle of radius $\frac{1}{2\pi}$
\begin{equation}\label{eq::circle}
    \text{dist}(x_1,x_2) = \frac{1}{2\pi} \text{min}\left( \abs{\theta_1 - \theta_2}, 2\pi -\abs{\theta_1 - \theta_2} \right),
\end{equation}
where $\theta_1, \theta_2$ are the polar angles of $x_1$ and $x_2$, respectively.


\subsection{Perimeter defense against point attacks}

{For ease of reference, the notation of this section is summarized in Table~\ref{tab:notation}.} Our problem is perimeter defense against point attacks with mobile defenders of varying speeds. Specifically, we have a perimeter $\cX$ in $d$-dimensional space, with a distance metric $\dist$, defended by $m$ mobile defenders with speeds $v_1, \dots, v_m$, so that defender $i$ at $x \in \cX$ at time $t$ can make it to $x'$ at time $t'$ iff
    \begin{align}
        \dist(x,x') \leq (t' - t) v_i
    \end{align}
Without loss of generality we order the defenders from fastest to slowest, i.e. $v_1 \geq \dots \geq v_m$, and we denote the \emph{speed vector} as $\bv = (v_1, \dots, v_m)$. Then $n$ attacks $(z_j,t_j) \in \cX \times \bbR_{\geq 0}$, where $z_j$ is the location on $\cX$ at which it happens, and $t_j$ is the time; WLOG we order these by time, i.e. $t_1 \leq \dots \leq t_n$. Because attacks happen at fixed locations and times, they cannot react to the positions of the defenders.


\begin{figure}
    \centering
    \includegraphics[trim={0cm 0.2cm 0cm 0.25cm},clip, width=0.9\textwidth]{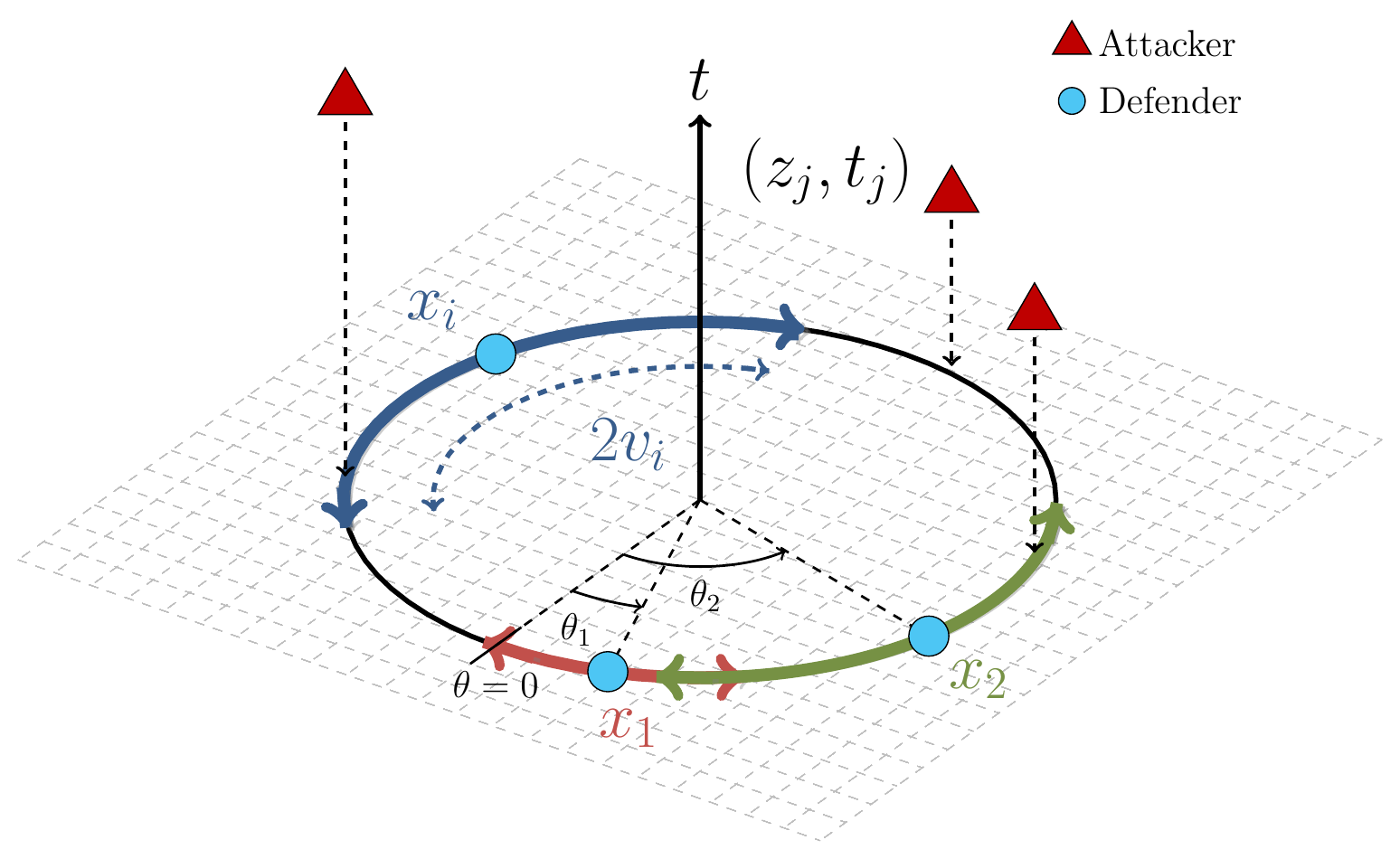}
    \caption{Three defenders facing three attacks, with the unit-time reachable sets for each defender shown. Note that the third dimension is time; if the attack represents a physical object it is approaching from somewhere outside the circle, but we are only concerned with where and when it will hit the perimeter. In this example the defenders are not allowed to leave the perimeter, so the size of the reachable set scales linearly with speed (until it covers the whole perimeter).}
    \label{fig:problem}
\end{figure}

An attack $(z_j, t_j)$ is \emph{thwarted} if and only if a defender is present, i.e. there is some defender $i$ at $z_j$ at time $t_j$; otherwise, we say that the attack \emph{breaches} the perimeter. The goal is to design a policy for the defenders that minimizes the number of attacks that breach the defenses, and to study the effectiveness of different defender speed combinations against attacks. 

Additionally, the team of defenders has a \emph{horizon} $h$ under which they can see attacks: specifically, at time $t$, any attack $(z_j, t_j)$ is known to the defenders if and only if $t_j \leq t + h$. We will study in particular the unit horizon $h = 1$ and the infinite horizon $h = \infty$ (all attacks are visible from the start).

Finally, the defenders are allowed to start at $t = 0$ at any combination of locations in $\cX$; they are even allowed to choose their starting locations based on the attack sequence (up to horizon $h$).

Given a speed vector $\bv$ and sequence of attacks $\{(z_j, t_j)\}_{j=1}^n$, we define $\opt(\bv, \{(z_j, t_j)\}_{j=1}^n)$ as the minimum number of attacks from $\{(z_j, t_j)\}_{j=1}^n$ that defenders of speed $\bv$ can let through (with all attacks known). In some cases we will be dealing with $\opt(\bv,\{(z_j, t_j)\}_{j=1}^n)$ for one sequence of attacks $\{(z_j, t_j)\}_{j=1}^n$ over many defender speed vectors $\bv$; in that case we write $\opt(\bv)$ for convenience.   

\subsection{Different settings}\label{subsec::settings}

Within the above problem description, there are several different variations, mostly to do with how the attacks are generated and the length of the horizon $h$. We roughly divide attack sequences into two settings:
\begin{enumerate}
    \item Any sequence of attacks $(z_j, t_j)$ is legitimate.
    \item Attacks must come at unit time intervals, i.e. $t_j = j$ for all $j \in [n]$.
\end{enumerate}
Note that in setting 2 we do not lose any generality by having the attacks happen at unit time intervals, since we can rescale the time units (and adjust the speeds of the defenders accordingly). Since the index $j$ is superfluous in setting 2 we refer to the sequence of attacks as $z_1, z_2, \dots, z_n$, indexed by $t$.

In setting 1, we study the case where all attacks are known to the defenders at the start; our primary problems are (i) find an algorithm for the defenders' movements that minimizes the number of breaches, and (ii) study the behavior of optimal defense against uniformly-random attacks (in both location and time) for different combinations of defenders. Since setting 1 is more general, the algorithms will also apply to setting 2.

In setting 2, we study the case where the attacks are (i) generated uniformly at random in location (time is fixed) and (ii) generated by an adversary which wants to guarantee a breach with as few attacks as possible. We also consider both the case where all the attacks are known to the defenders at the start ($h = \infty$) and the case where attack $t$ only becomes known at time $t-1$ ($h = 1$).






\begin{remark}
Here we deal with the case where the number of defenders is fixed, and the question is how fast to make each defender (and in particular whether to make them all the same speed or not). The alternative case of varying the number of defenders is investigated in the Appendix, especially in regards to the tradeoff between fewer and faster defenders versus more and slower ones.
\end{remark}

\section{Infinite Horizon Theoretical Results}\label{sec::theoretical_results}

\subsection{Dynamic programming with infinite horizon}

\begin{figure}
    \centering
    \includegraphics[trim={0cm 0.35cm 0cm 0.2cm},clip,width=\textwidth]{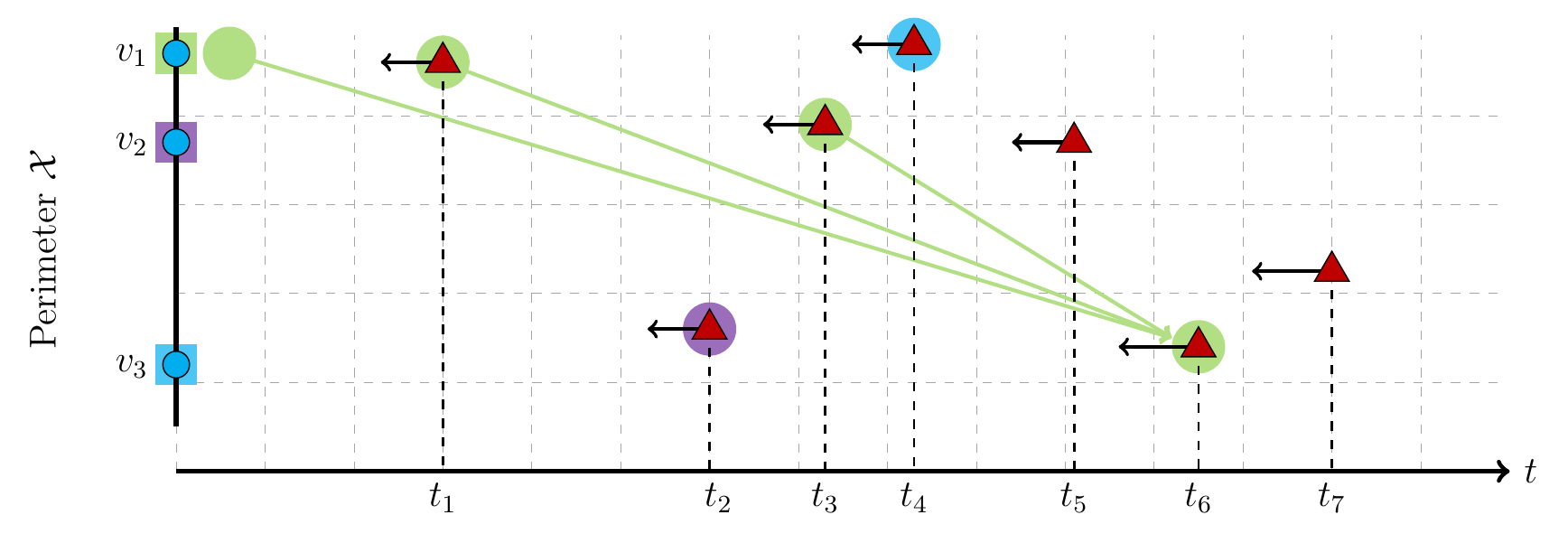}
    \caption{Computing $f(6, 2, 4)$ (defender $1$ has to thwart attack $6$, etc.) recursively; each defender is allowed to thwart attacks prior to these, but not afterwards. Since $6$ is the maximum value, we consider the last attack that defender $1$ can handle before $6$: based on its speed, it can be $0$ (defend nothing before $6$), $1$, or $3$. Thus $f(6, 2, 4) = \min(f(0,2,4), f(1,2,4), f(3,2,4)) - 1$.}
    \label{fig:dynamic}
\end{figure}

We now give an algorithm which, given defender speeds $\bv = (v_1, \dots, v_m)$ and attacks $\{(z_j,t_j)\}_{j=1}^n$ returns two things: (i) $\opt(\bv, \{(z_j,t_j)\}_{j=1}^n)$ (the minimum number of attacks that can be let through); and (ii) the list (of lists) $\boldsymbol{\ell} = (\ell_1, \dots, \ell_m)$, where $\ell_i$ is the (sub)sequence of attacks which defender $i$ should thwart. We refer to $\boldsymbol{\ell}$ as a \emph{defense plan}.

Recall that by default the attacks are sorted in order of arrival time (or the user should sort them before applying the algorithm). 

The pseudocode is given in Alg.~\ref{alg:two}, in which we use the following notation: $\bj = (j_1, \dots, j_m) \in \{0,1,\dots,n\}^m$ denotes a vector of attacks assigned to each defender (with $j_i = 0$ indicating no attack assigned to defender $i$, and we allow the $j_i$'s to be non-distinct even though it is redundant); 
\begin{align}
    \bj_{-i}(j') = (j_1, \dots, j_{i-1}, j', j_{i+1}, \dots, j_m)
\end{align}
i.e. $\bj$ with the $i$th entry replaced by $j'$;
\begin{align}
    \bone_{i}(j', j'') := \begin{cases} 1 &\dist(z_{j'},z_{j''}) \leq (t_{j''}-t_{j'}) v_i \\ 0 &\text{otherwise} \end{cases}
\end{align}
is the indicator that defender $j$ is capable of thwarting attack $j''$ after thwarting $j'$ (and $\bone_{i}(0, j'') = 1$ since defenders can start anywhere); $[\cdot] + [\cdot]$ denotes concatenation (of lists); and $\argmin$ ($\argmax$) denote the sets of values minimizing (maximizing) the arguments. The for-loop in Alg.~\ref{alg:two} iterates in lexicographic order, skipping $f(0,\dots,0)$ (whose value is already known) so that the recursion can work.

The proof of the following result is contained in the Appendix:
\begin{theorem} \label{thm::alg:two-works}
Alg.~\ref{alg:two} outputs the correct value of $\opt(\bv, \{(z_j,t_j)\}_{j=1}^n)$ and $\boldsymbol{\ell}$.
\end{theorem}

\begin{remark} 
Alg.~\ref{alg:two} relies on the subtle point that $i^* \in \argmax_i \, j_i$ because if not, then we do not know whether to subtract $1$ when we do the update; by setting $i^* \in \argmax_i \, j_i$, we remove the question of whether a defender $i'$ assigned to a later $j_{i'}$ can also thwart attack $j_{i^*}$.
\end{remark}

\begin{remark}
Alg.~\ref{alg:two} assumes that the defenders can start at whatever locations they want, but can be modified for fixed defender starting locations (or a set of possible starting locations) by redefining $\bone_i (0, j)$ to indicate whether they can reach attack $j$ from their starting locations. It can also be modified for the important case where attacks cause varying amounts of damage, with attack $j$ doing $w_j$ damage (should it not be intercepted); see for instance the Iron Dome missile defense system, which prioritizes attacks based on potential damage estimates \cite{military}. To make this modification, replace $-1$ with $-w_{j_{i^*(\bj)}}$ in line 7 and $f(0,\dots,0) = c = n$ with $f(0,\dots,0) = c = \sum_j w_j$.
\end{remark}

\vspace{-1cm}
\begin{algorithm}
\caption{Dynamic programming for infinite horizon defenders.}\label{alg:two}
\KwData{Attacks $\{(z_j, t_j)\}_{j=1}^n$ ; defender speeds $\bv = (v_1, \dots, v_m$)}
\KwResult{$\opt(\bv,\{(z_j, t_j)\}_{j=1}^n)$}
$f(0, \dots, 0) = n$,~ $c = n$,~ $\bj^{\min} = (0, \dots, 0)$ \tcc*[r]{initialization} 
\BlankLine
\For(\tcc*[f]{compute f}){$\bj \in \{0,1,\dots,n\}^m$}{
 $\cS = \argmax_i \, \{j_i\}$\;
 Choose $i^*(\bj) \in \cS$\;
 Choose $j^*(\bj) \in \argmin_{j'} \, \big\{ f(\bj_{-i^*(\bj)}(j')) : j' < j_{i^*(\bj)} \text{ and } \bone_{i^*}(j', j_{i^*(\bj)})  \big\}$ 
 
  \eIf{$\abs{\cS} = 1$}{
$f(\bj) = f(\bj_{-i^*(\bj)}(j^*(\bj))) - 1$
  }{
  $f(\bj) =  f(\bj_{-i^*(\bj)}(j^*(\bj)))$
  }
  \If{$f(\bj) < c$}{
  $c = f(\bj)$, ~$\bj^{\min} = \bj$
  }
}
\BlankLine
$\boldsymbol{\ell} = (\ell_1, \dots, \ell_m) = ([j^{\min}_1], \dots, [j^{\min}_m])$ \tcc*[r]{initialize defender lists}

$\bj^{\text{curr}} = \bj^{\min}$

\BlankLine

\While(\tcc*[f]{reconstruct defender lists}){$\bj^{\text{curr}} \neq (0, \dots, 0)$}{
\If{$j^*(\bj) \neq 0$}{
$\ell_{i^*(\bj)} = [j^*(\bj)] + \ell_{i^*(\bj)}$ \tcc*[r]{add $j^*(\bj)$ to front of list}
}

$ j^{\text{curr}}_{i^*(\bj)} = j^*(\bj)$
}

\BlankLine

\Return{$\opt(\bv, \{(z_j,t_j)\}_{j=1}^n) = c$, ~$\boldsymbol{\ell}$}

\end{algorithm}
\vspace{-0.5cm}
Given $m$ defenders and $n$ attackers, the number of computations needed to run Alg.~\ref{alg:two} is on the order of $(n+1)^{m+1}$ (we need to run through $(n+1)^m$ values of $\bj$, and each update takes up to $n$ time for the comparisons).

\subsection{Monotoniticy-based computational acceleration}\label{sec::monotonicity}

In order to investigate team heterogeneity, we compute $\opt(\bv, \{(z_j, t_j)\}_{j=1}^n)$ for all $\bv$ whose elements $v_i$ are at $g$ evenly-spaced locations in a range $(v_{\min}, v_{\max}]$.\footnote{For instance, if $g = 5$ and $(v_{\min}, v_{\max}] = (0,1]$, we measure $\bv$ where $v_i \in \{0.2, 0.4, 0.6, 0.8, 1\}$ for all $i$.} We refer to $g$ as the number of \emph{grains}. If we were to run Alg.~\ref{alg:two} for all combinations $\bv$ of speeds, the complexity becomes $O((n+1)^{m+1} g^m)$, which gets extremely large very quickly.

However, as each attack sequence is evaluated on all $\bv$, we can take advantage of the monotonicity of $\opt$ over $\bv$ to reduce the amount of computation needed. 

In particular, for any sequence $\{(z_j, t_j)\}_{j=1}^n$,
\begin{align}
    \bv \leq \bv' \implies \opt(\bv) \geq \opt(\bv') 
\end{align}
since faster defenders can always emulate slower ones and thus achieve (at least) as good a result on any attack sequence. This means that
\begin{align}
    &\opt(\bv) = \opt(\bv') = k \text{ for some } \bv \leq \bv'
    \\ \implies &\opt(\bv'') = k \text{ for all } \bv \leq \bv'' \leq \bv' \, .
\end{align}
Thus we know $\opt(\bv'') = k$ for a range of $\bv''$, without having to run Alg.~\ref{alg:two}. Taking the set of values $\bv \in (v_{\min}, v_{\max}]^m$ (of given grains), for any $\{(z_j, t_j)\}_{j=1}^n$ we can evaluate $\opt(\bv, \{(z_j, t_j)\}_{j=1}^n)$ in a strategic order to minimize the number of times we need to run Alg.~\ref{alg:two}. This is discussed in greater detail in the Appendix.

\section{Unit Horizon Theoretical Results} \label{sec::unit-horizon}

This section considers defenders with a unit horizon of incoming attacks. The general setup is 
\begin{itemize}
\item We consider two defenders with speeds $v_1 \geq v_2$.
    \item We consider a perimeter $\cX$ homeomorphic to $\cS^1$ (a circle\footnote{We consider this case because it has a number of nice symmetries, and because perimeters enclosing a simply-connected 2D area are homeomorphic to $\cS^1$.}), with distances determined by arc length and total length normalized to $1$; we represent $\cX = [-1/2, 1/2]$ (but $-1/2$ and $1/2$ are the same point). To denote this situation, we define the distance function
    \begin{align}
        \text{dist}(y_1, y_2) = \min \big\{|y_1 - y_2|, 1 - |y_1 - y_2| \big\}
    \end{align}
    (a rescaled version of \eqref{eq::circle}).
    The maximum possible value of $\text{dist}(y_1, y_2)$ is $1/2$, and we assume they start at maximum distance from each other, i.e., at antipodal points. 
    \item The $n$ attackers are generated according to Setting 2 from \Cref{subsec::settings}: attacker $t$ appears at time $t$, uniformly (and independently) over $\cX$.
    \item The defenders have a unit horizon in time: at any given time they only see the next incoming attack, though they also know $n$ and the current time $t$.

\end{itemize}
Therefore the defenders' policy can be thought of as a sequence of decisions taken at unit time intervals (i.e. when the next attack is revealed), which is naturally formulated as a Markov Decision Process (MDP) \cite{puterman2014markov} with $n$ steps, with the reward being the number of thwarted attacks.

To simplify the MDP we can remove one state variable since, by symmetry, we can rotate $\cX$ (or relabel it) so that at the beginning of any time step, defender 1 is at location $0$. We can also reflect it so that defender 2 is on the positive half. Thus the state at time $t$ (just before the location of the next attack is revealed) can be denoted by a single parameter $a(t)$, indicating the distance between the two defenders. Then the next attack's location $x(t+1)$ is revealed, in the coordinate system relative to the defenders' positions.

\subsection{Policy and Reward}

A \emph{unit-horizon policy} is a function $f: [0,1/2] \times [-1/2,1/2] \to [0,1/2]$. The inputs are $a(t)$, $x(t)$ and the number of remaining attacks, and the output is $f(a(t),x(t)) = a(t+1)$. As described above, $a(t+1)$ is the distance between the two defenders at time $t+1$. $f$ must satisfy the condition
\begin{align}
    f(a(t), x(t)) \in [a(t) - v_2 - v_1, a(t) + v_2 + v_1]
\end{align}
The policy then produces a reward
\begin{align}
    r(t) := r(a(t), x(t), f(a(t), x(t))) \in \{0,1\}
\end{align} 
the reward, based on whether the given movement makes it possible for the attack to be thwarted ($r(t) = 1$ if so, $=0$ if not). $r(t)$ is given as follows:
\begin{align*}
    r(t) = \begin{cases} 1 & \text{if } \dist(0, x(t)) \leq v_1 \text{ and } [\text{dist}(x(t), a(t)) - v_2, \text{dist}(x(t), a(t)) + v_2] \\ 1 &  \text{if } \dist(a(t),x(t)) \leq v_2 \text{ and } f(a(t), x(t)) \in [x(t) - v_1, x(t) + v_1] \\ 0 & \text{otherwise} \end{cases}
\end{align*}
The reason for this is that by symmetry (of the perimeter and of the attacks), given the distance $a(t+1) = f(a(t), x(t))$ between the defenders at the start of the next step, the ability of the defenders to stop future attacks does not depend on their locations. Thus, if the defenders can stop the current attack and end at distance $a(t+1) = f(a(t), x(t))$ for the next step, this is always preferable to ending at the same distance \emph{without} making the capture. 

Hence $r(t) = 1$ under policy $f$ if and only if this is possible, which can be split into two cases: (i) defender $1$ makes the capture; (ii) defender 2 makes the capture. If either of these are feasible, $r(t) = 1$; if neither are, $r(t) = 0$.



\begin{remark}
If $\text{dist}(a(t), x(t)) > {v_2}$ and $\text{dist}(0, x(t)) > {v_1}$, this means that neither defender can reach the next attack and hence $r(t) = 0$ no matter what.
\end{remark}

\subsection{Optimal defender policy}
Fix a defender policy $f$. For a given total number $N$ of incoming attacks and an initial distance $a$ between the two defenders, we define the expected reward $J(a; N)$ of the defenders as the expected total number of thwarted attacks, i.e.,

\begin{equation}
J(a; N) := \mathbb{E}_x \left[ \sum_{t=0}^{N-1} r(t) \right] \text{ under policy } f,
\end{equation}
where the expectation is over the attack locations $x(t)$. With this definition, we are interested in determining the policy $f$ that leads to the highest expected reward. We show in the Appendix that for a wide range of values for $v_1, v_2$ and $N$, the optimal strategy should $(i)$ always thwart the currently-known if possible. We next prove that the optimal policy subsequently should $(ii)$ always maximize $a(t)$ subject to the first constraint. That is:

\begin{proposition} \label{prop::maximize-distance}
$f^*$ maximizes $J(a; N)$ if (ii) $a(t+1)$ is maximized for all inputs, over all policies that satisfy $(i)$ (i.e. capture when possible).
\end{proposition}

We next show necessary and sufficient conditions for perfect defense, i.e. when no (fixed-time) attack sequence can force a breach.
\begin{theorem}[The perfect defense theorem]
For any pair of defenders with speeds $v_1, v_2$ where $v_2 \leq v_1$, there exists a sequence of attacks that breaches if and only if $v_1 < 1/2$ and $v_1 + 3v_2 < 1$. Furthermore, if $v_1 \geq 1/2$ or $v_1 + 3v_2 \geq 1$, the defenders can defend indefinitely even with a one-step horizon. Furthermore, if any sequence of attacks guarantees a breach, there is a sequence of at most $6$ attacks that does so.
\end{theorem}

Both proofs are given in the Appendix.

\section{Simulation Results} \label{sec::simulation_results}

We conduct simulations for each of the settings from \Cref{subsec::settings}. Our experiments are run as follows: 
\begin{enumerate}
    \item Generate attacks $\{(z_j, t_j)\}_{j=1}^n$ randomly, either with fixed attack times $t_j = j$ or uniformly-random attack times in $[0,t_{\max}]$.
    \item Compute $\opt(\bv, \{(z_j, t_j)\}_{j=1}^n)$ for $\bv \in (v_{\min}, v_{\max}]^m$, at $g$ intervals.
    \item Repeat the above for $T$ trials and average the resulting values for each $\bv$.
\end{enumerate} 

We conduct all of our experiments on a circular perimeter of circumference $1$, where the defenders are not permitted to leave the perimeter (so maximally distant points are at opposite ends and have distance $1/2$). Comparison of the results sheds light on the conditions which favor heterogeneous defender teams and those which favor homogeneous teams and/or single super-defenders.

 
The structure of the simulations means each combination of defender speeds is evaluated on the same set of attack sequences, which makes the comparison fairer, and allows us to significantly speed up the computation when evaluating $\opt(\bv, \{(z_j, t_j)\}_{j=1}^n)$ for many values of $\bv$ on a single attack sequence $ \{(z_j, t_j)\}_{j=1}^n$, by exploiting the fact that $\opt$ is a monotonically-decreasing step function in $\bv$ (as described in \Cref{sec::monotonicity}).

The full list of parameters is given in \Cref{tab:parameters}.

\begin{table}[!t]
	\centering
	\caption{Parameters of the experiments}
	\label{tab:parameters}
	\begin{tabular}{@{}ll@{}}
		\toprule
		Symbol         & Description                    \\ \midrule
		$m$          & Number of defenders ($m=2$ unless specified otherwise)              \\
		$n$          & Number of attacks               \\
		$T$ & Number of trials \\
		$t_{\max}$ & Size of attack window (not needed for heterogeneous setting (ii)) \\
		$(v_{\min}, v_{\max}]$ & Range of defender speeds (inclusive of $v_{\max}$ but not $v_{\min}$) \\
		$g$ & Number of speed values measured (grains) within $(v_{\min}, v_{\max}]$ \\
		\bottomrule
	\end{tabular}
\end{table}

\subsection{Simulation Results}


In \Cref{fig:heterogeneous-many-trials}, we simulate sequences of $n = 25$ attacks of both settings, where the perimeter $\cX$ is a unit circle of circumference $1$ and $m = 2$ defenders; for uniformly random attack times we set $t_{\max} = 25$ to get the same density of attacks in both cases. This is analyzed over the speed range $(v_{\min}, v_{\max}] = (0, 0.6]$ with $g = 256$ grains. The left column shows results for uniformly-random attack times; the right column shows results for fixed attack times.

The results are given as surface plots, taking defender speeds $v_1, v_2$ and returning $\opt(v_1, v_2)$ (ignoring the assumption in the analysis that $v_1 \geq v_2$, so the plots are symmetric about the line $v_1 = v_2$). We give:

\begin{itemize}
    \setlength{\itemindent}{0em}
    \item {\bf Top row:} $\opt(v_1, v_2)$ for a single sequence of attacks. This can be viewed as $T = 1$, and is meant to give a visualization of how adjusting the speeds of the defenders changes the ability to defend against a particular sequence. Since $\opt(v_1, v_2)$ takes integer values, we have a monotonically-decreasing step function.
    \vspace{0.2pc}
    \item {\bf Middle and bottom rows:} $\opt(v_1, v_2)$ when averaged over $T = 200$ randomly-generated attack sequences. Middle row gives the front view to show overall shape; bottom row gives the back view to show the ridge at $v_1 = v_2$. This ridge, which appears for both uniformly-random attack times and fixed attack times, shows that on average homogeneous defenders are less efficient (per combined speed) than heterogeneous defenders.
\end{itemize}

\begin{figure}[!ht]
     \centering
     \begin{subfigure}[b]{0.48\textwidth}
        \centering
        \includegraphics[width=\textwidth]{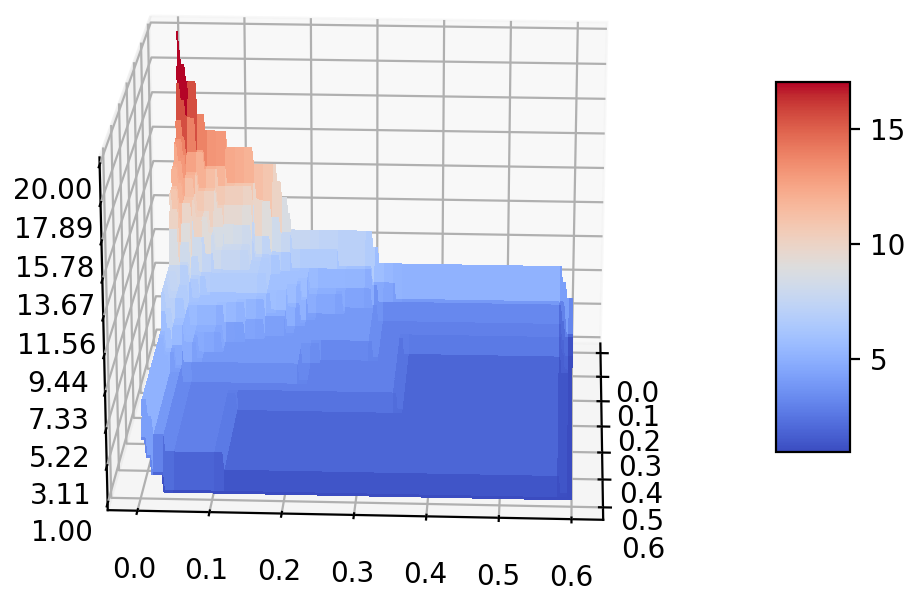}
    \end{subfigure}
    \hfill
    \begin{subfigure}[b]{0.48\textwidth}
        \centering
        \includegraphics[width=\textwidth]{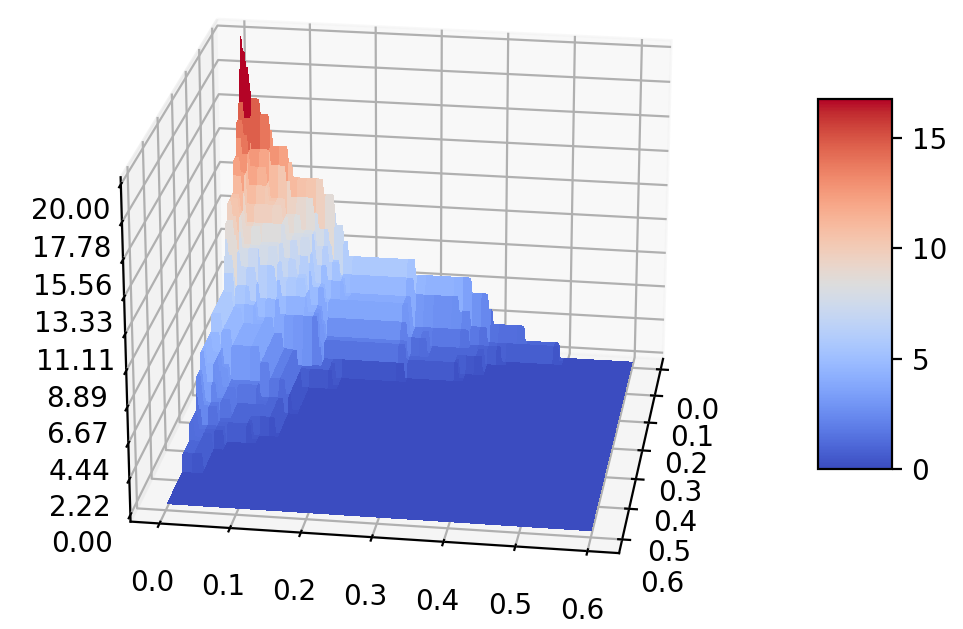}
    \end{subfigure}

    \begin{subfigure}[b]{0.48\textwidth}
        \centering
        \includegraphics[width=\textwidth]{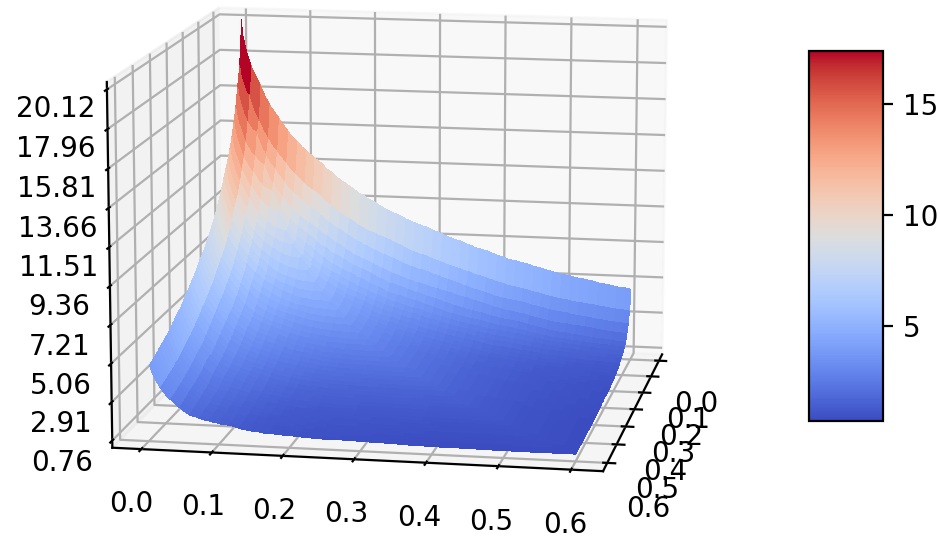}
    \end{subfigure}
    \hfill
    \begin{subfigure}[b]{0.48\textwidth}
        \centering
        \includegraphics[width=\textwidth]{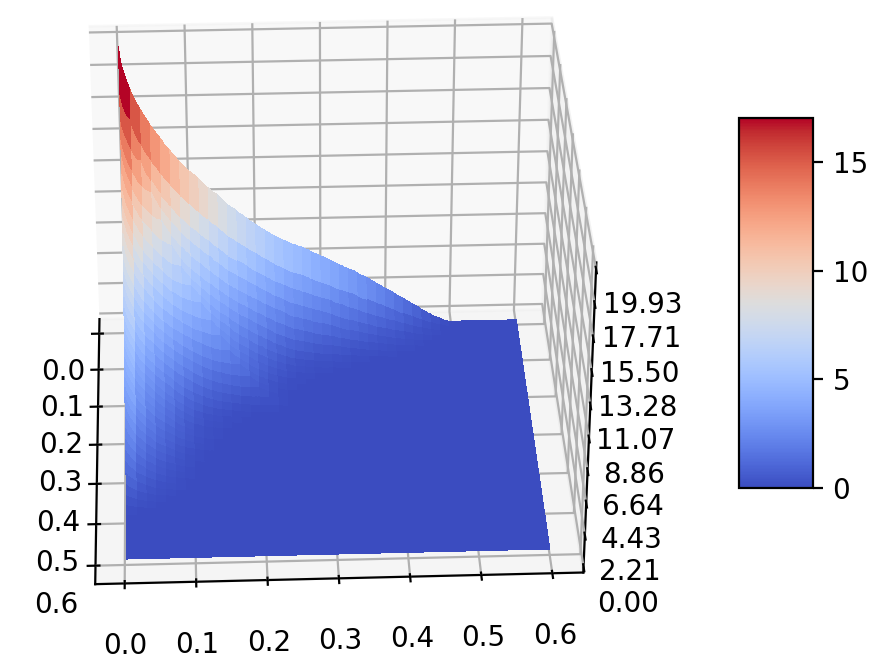}
    \end{subfigure}
    
    \begin{subfigure}[b]{0.48\textwidth}
        \centering
        \includegraphics[width=\textwidth]{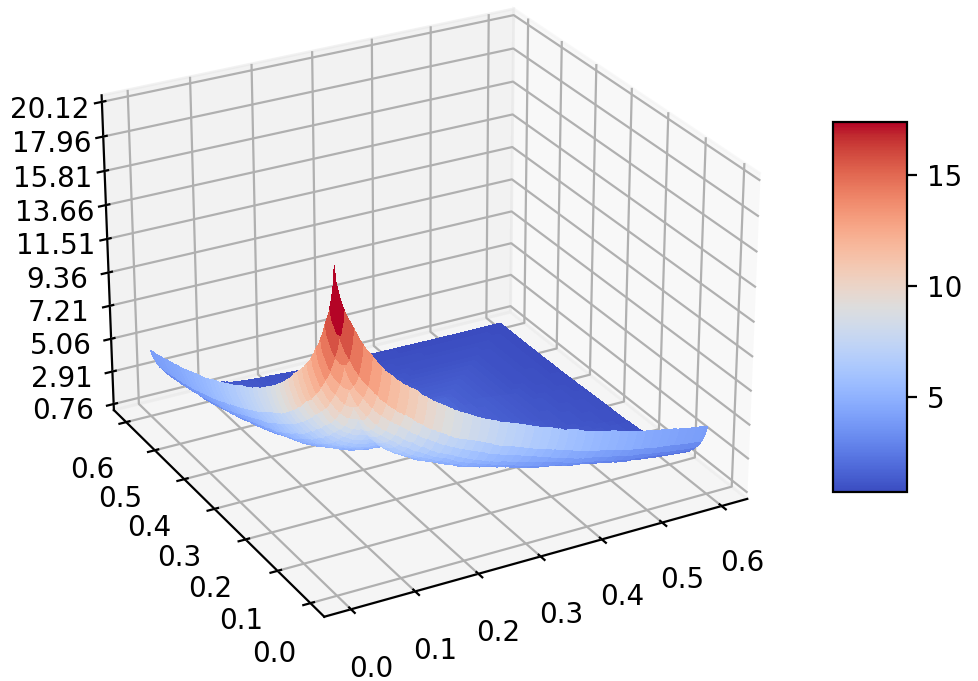}
    \end{subfigure}
    \hfill
    \begin{subfigure}[b]{0.48\textwidth}
        \centering
        \includegraphics[width=\textwidth]{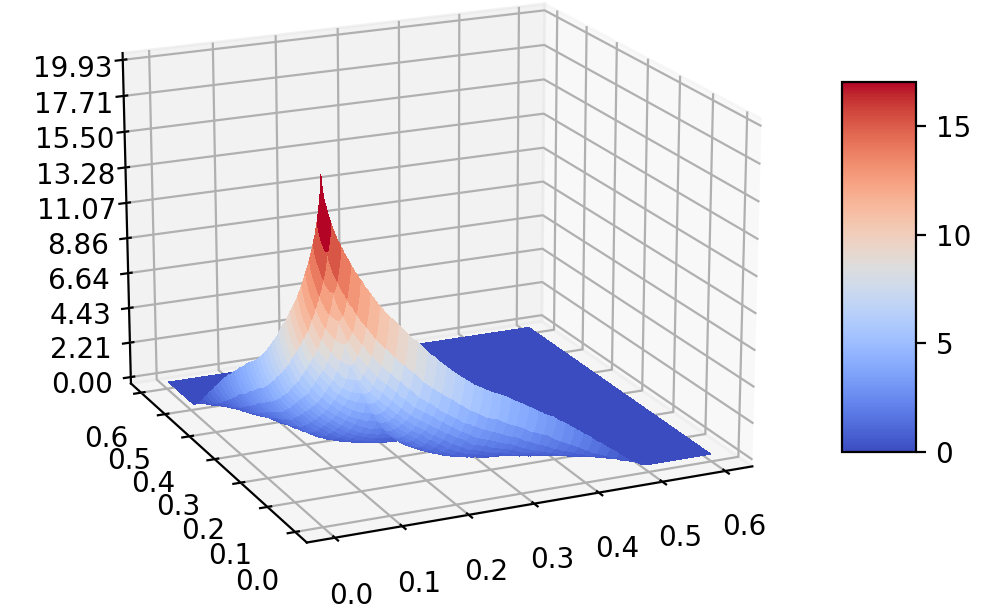}
    \end{subfigure}
        \caption{$2$ defenders evaluated at $g = 256$ grains for speeds $(v_{\min}, v_{\max}] = (0,0.6]$ for $200$ trials. \emph{Top row:} 1 trial, front view. \emph{Middle row:} 200 trials, front view. \emph{Bottom row:} 200 trials, back view, showing the `ridge' at the center line $v_1 = v_2$ (for both attack types). \\ {\bf Left:} Uniformly-random attack times, $n = 25$, $t_{\max} = 25$
        \\ {\bf Right:} Fixed attack times, $n = 25$}
        \label{fig:heterogeneous-many-trials}
\end{figure}


\noindent From this we can make a number of interesting observations:
\begin{itemize}
    \item $\opt(v_1, v_2)$ is generally larger for the uniformly random attack times, as attacks which are close together in time are much harder to defend. In particular, with fixed attack times $\opt(v_1, v_2) = 0$ for sufficiently large defender speeds (one defender of speed $1/2$ is 
    \item As mentioned, there is a ridge on $v_1 = v_2$ (the back view makes it clearly visible). This shows that on average, homogeneous defenders are less effective than well-designed heterogeneous ones.
    \item Under uniformly-random attack times, each `half' (cutting at the $v_1 = v_2$ line) is empirically convex, while under fixed attack times, each `half' is convex near the $v_1 = v_2$ ridge but becomes concave again near the edge of the plot (as seen in the back view) and as the defender speeds increase (as can be seen on the edge in both views). 
\end{itemize}
We also consider the question: what is the optimal \emph{mix} of defender speeds? To answer this, we need to consider what we want to hold constant, since obviously faster defenders are always better; an obvious starting point is to look at defenders of a fixed total speed, and consider what ratio of speeds performs the best. This also means that we are comparing defender teams whose reachable sets are of equal total size (ignoring overlaps), and (because we evaluate over a grid of possible values of $\bv$) means we compare the values of $\opt(\bv)$ on a diagonal line.

In \Cref{fig:heterogeneous-best-mix}, we show the best (empirical) mixture: for each value of $v_{tot} = v_1 + v_2$, the returned value is
\begin{align}
    \frac{v_2^*}{v_{tot}} \text{ where } v_2^* := \argmin_{v_2 \leq v_{tot}/2} \opt((v_{tot}-v_2, v_2))
\end{align}
That is, given a total speed of $v_{tot}$, what is the optimal fraction of the speed `budget' to assign to the slower defender? A value of $0.5$ signifies homogeneous defenders are best; a value of $0.0$ signifies that a single super-defender is best; and a value in between signify some heterogeneous mix of defenders is best.

\begin{figure}[!ht]
     \centering
        \includegraphics[width=\textwidth]{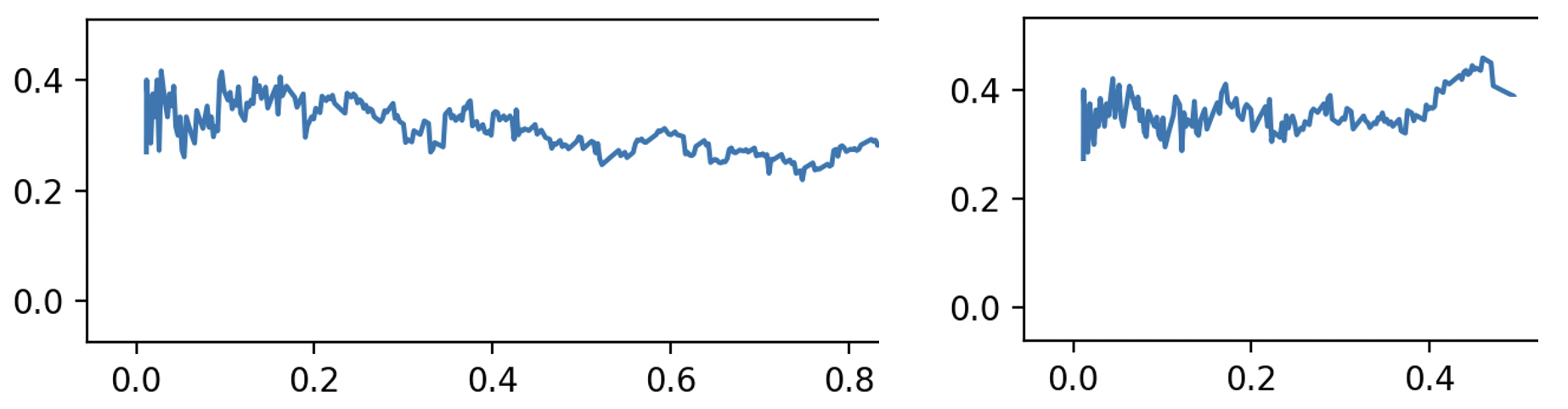}
        \caption{Empirical optimal ratio $v_2/v_{tot}$, for various values of $v_{tot}$. \emph{Left:} Uniform attack times. \emph{Right:} Fixed attack times.}
        \label{fig:heterogeneous-best-mix}
\end{figure}
These are based on the same experiments as shown in \Cref{fig:heterogeneous-many-trials}. Note that the fixed attack times graph ends at $v_{tot} = 0.5$; past that, both one single super defender and homogeneous defenders will defend perfectly, so measuring the minimum no longer makes sense. However, it is striking that the benefits of a heterogeneous team persist so close to that threshold, and the optimal ratio remains relatively stable over a wide range of speed `budgets' in both settings.

\paragraph{Computational complexity of simulations:}  



The results of the monotonicity-based computational acceleration discussed in \Cref{sec::monotonicity} can be seen in \Cref{fig:monotonicity-speedup}, corresponding to the simulations shown in \Cref{fig:heterogeneous-many-trials}. As before, the left-hand column is the results for uniformly-random time attacks, and the right-hand column is the results for fixed time attacks, while the top row represents a single trial (corresponding to the top row of \Cref{fig:monotonicity-speedup}) and the bottom row correspond to the average of $T = 200$ trials. 

Each square is a $256 \times 256$ grid, representing the $256^2$ combinations of speeds $\bv$ for which we want to compute $\opt(\bv)$; the shade of a given point represents the fraction of times Alg.~\ref{alg:two} had to be run on for that specific $\bv$ (as opposed to the value being known already due to monotonicity), running from yellow (Alg.~\ref{alg:two} never had to be run) to purple (Alg.~\ref{alg:two} always had to be run). Note that because they represent a single trial (each), every point in the top two graphs takes a value of either $0$ or $1$.


We note a few things: (i) the savings increase strongly where $\bbE [\opt(\bv)]$ is flatter (this is expected since $\nabla_{\bv} \bbE [\opt(\bv)]$ corresponds to the probability that there is a step at $\bv$, and having a step nearby means the condition is less likely to be satisfied); (ii) there are darker points at regular intervals (such as in the center), which correspond to the combinations which are evaluated earlier.

Even with $m = 2$ and the strategic use of monotonicity, which can save up to about 95\% of the running time, this can get big fairly quickly.

\begin{figure}[!ht]
    \centering
     
    \begin{subfigure}[b]{0.40\textwidth}
        \centering
        \includegraphics[width=\textwidth]{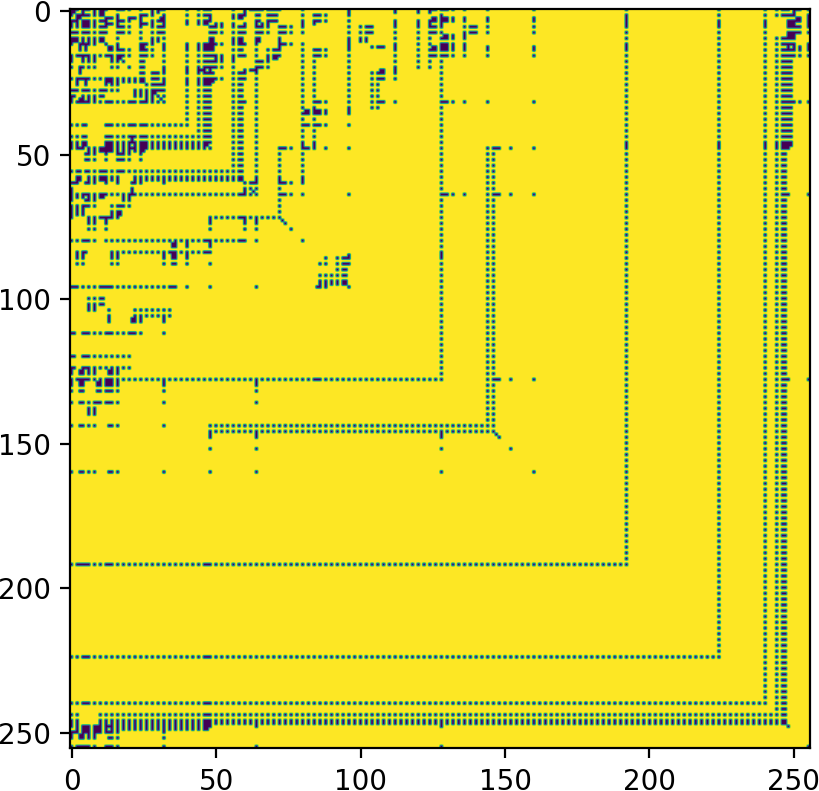}
    \end{subfigure}
    \hfill
    \begin{subfigure}[b]{0.40\textwidth}
        \centering
        \includegraphics[width=\textwidth]{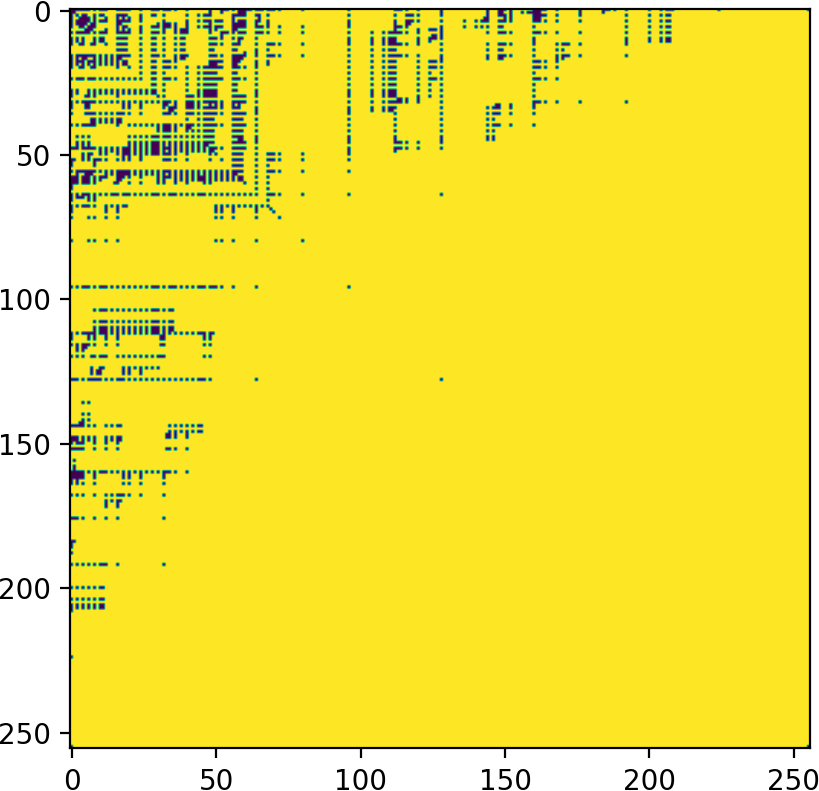}
    \end{subfigure}
     
     \vspace{0.5pc}
     
    \begin{subfigure}[b]{0.40\textwidth}
        \centering
        \includegraphics[width=\textwidth]{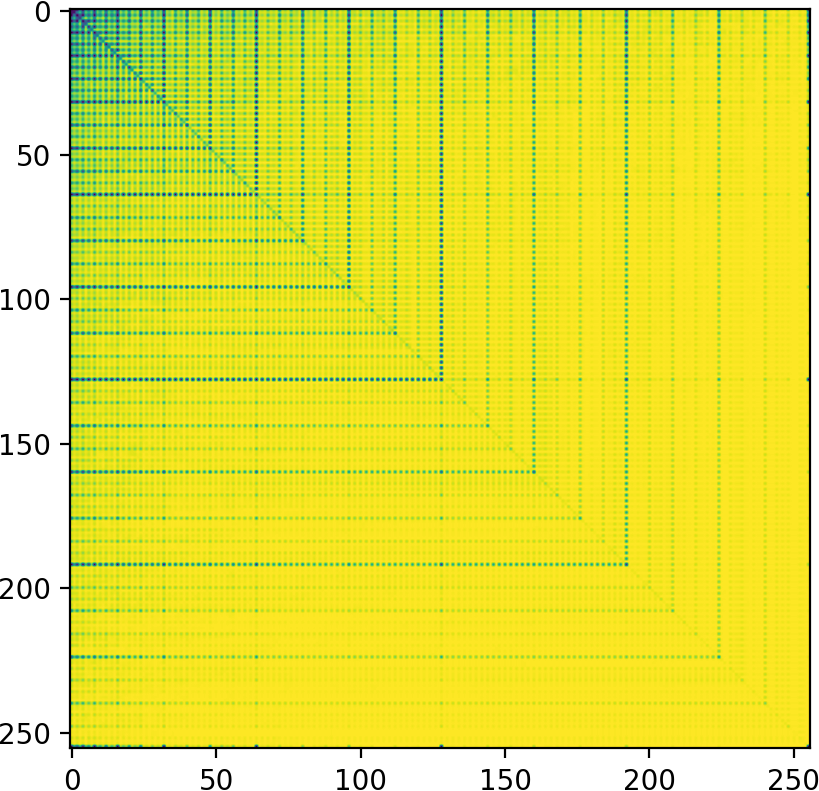}
    \end{subfigure}
    \hfill
    \begin{subfigure}[b]{0.40\textwidth}
        \centering
        \includegraphics[width=\textwidth]{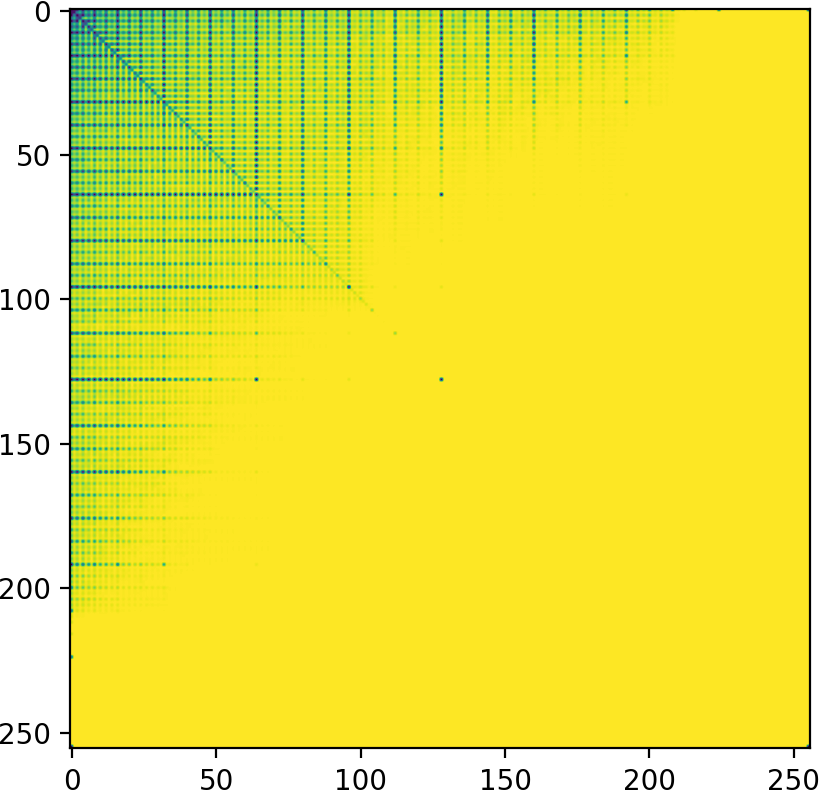}
    \end{subfigure}
     
    \caption{Monotonicity savings for the trials depicted in \Cref{fig:heterogeneous-many-trials}. Uniformly-random attack times on the left, and fixed attack times on the right. Axes labeled by position in the vector of possible speeds ($0$ to $g-1$). Top row is for one trial (corresponding to the single trials shown in \Cref{fig:heterogeneous-many-trials}) and bottom is average over 200 trials.}
        \label{fig:monotonicity-speedup}
\end{figure}

\subsection{Simulations for unit horizon}
Simulation results for the case of two defenders on a circular perimeter with unit horizon are shown in Figure~\ref{fig:unit_results}. Note that in this case, heterogeneity is not beneficial, it is even detrimental. The optimal speed allocation is to assign the entire speed budget to one defender or split it equally.

\begin{figure}[!ht]
     \centering
    \begin{subfigure}[b]{0.45\textwidth}
        \centering
        \includegraphics[trim={1.9cm 1cm 1.5cm 2cm},clip,width=\textwidth]{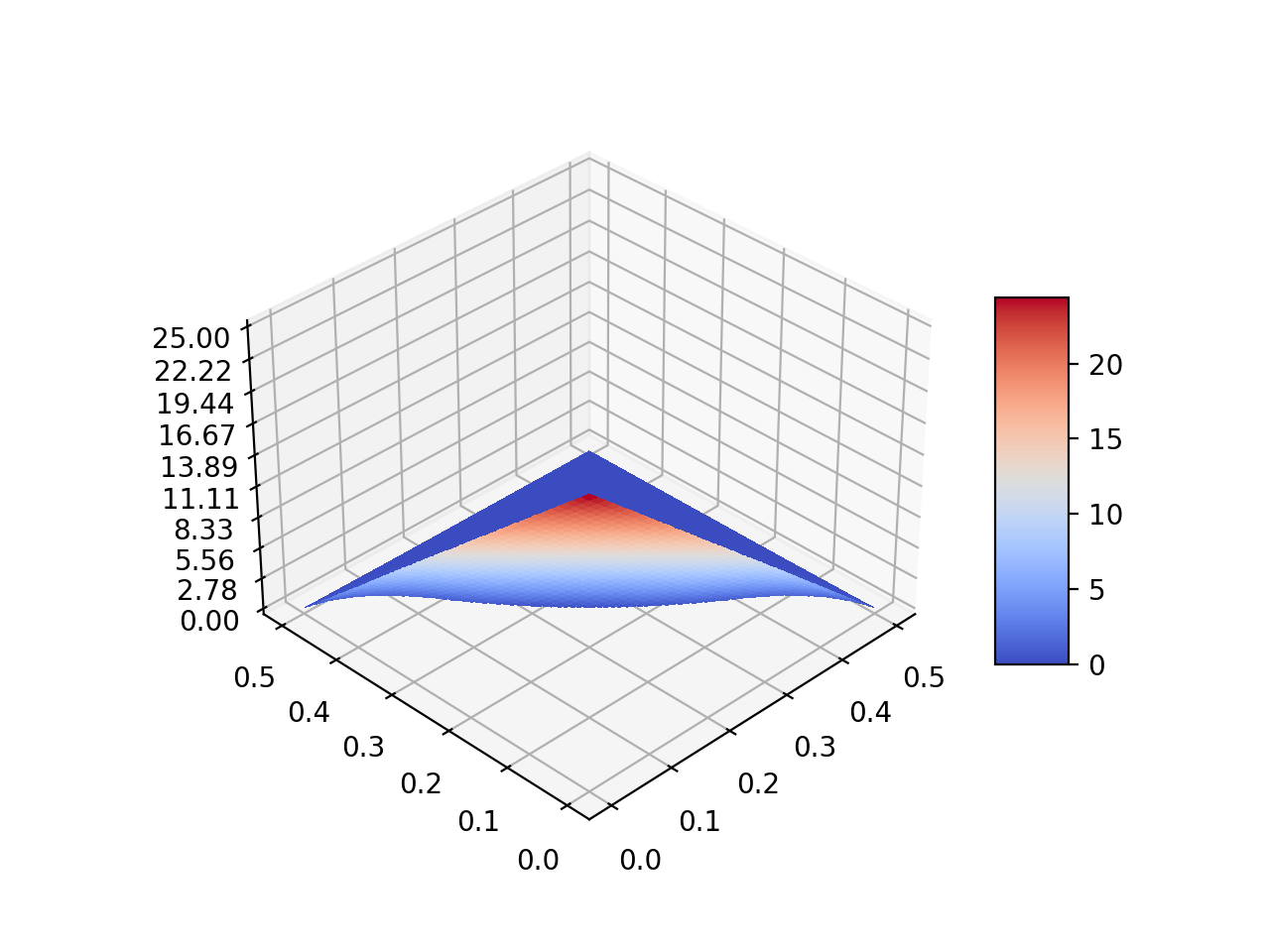}
    \end{subfigure}
    \hfill
    \begin{subfigure}[b]{0.45\textwidth}
        \centering
        \includegraphics[trim={2.5cm 1.5cm 1.8cm 2cm},clip,width=\textwidth]{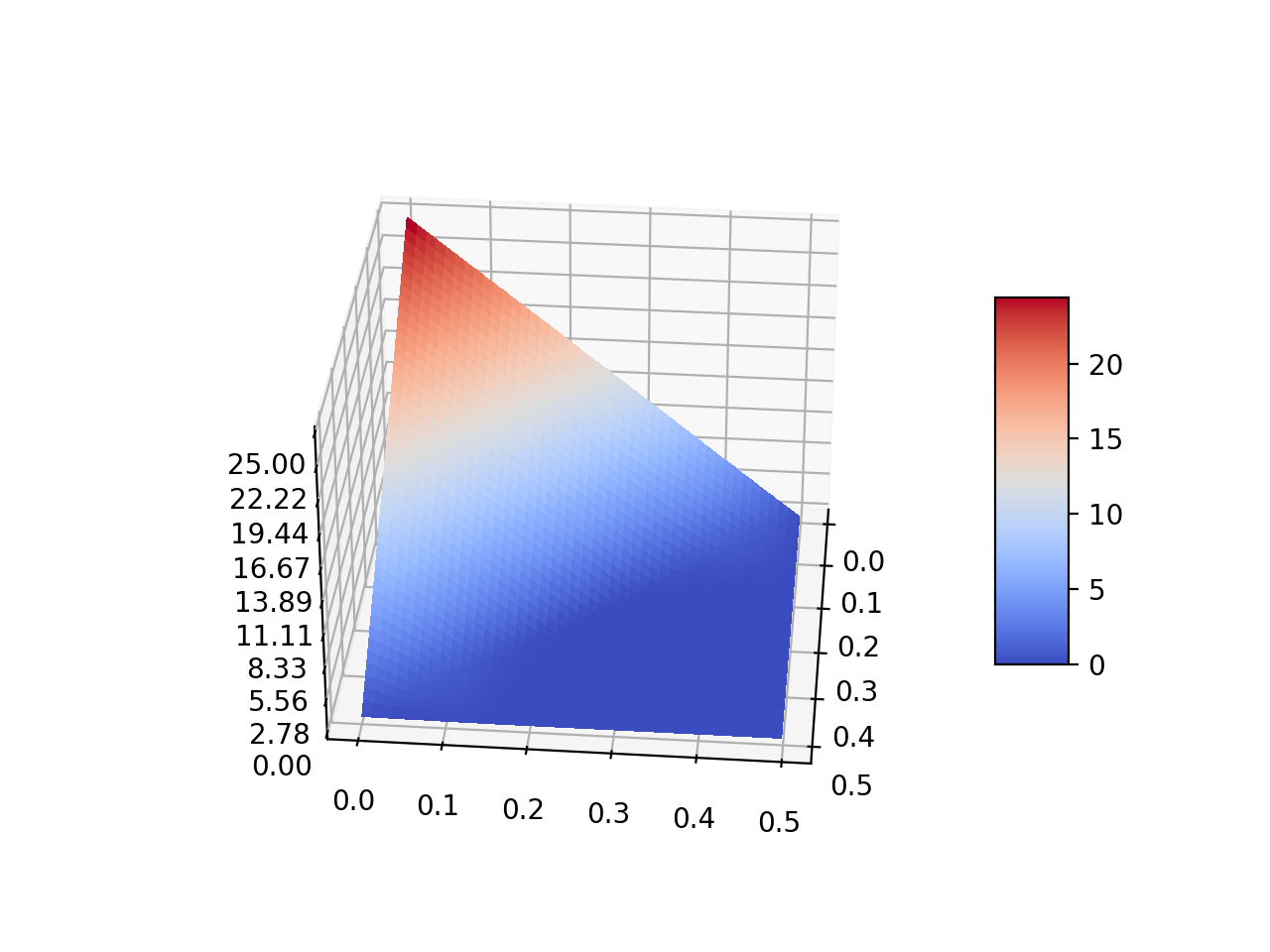}
    \end{subfigure}
    
 \caption{Unit-horizon case, $2$ defenders evaluated at $g = 128$ grains for speeds $(v_{\min}, v_{\max}] = (0,0.5]$ for $10000$ trials and $n=25$ attacks. \emph{Left:} back view, note the lack of the `ridge' seen in \Cref{fig:heterogeneous-many-trials}. \emph{Right:} front view.}
        \label{fig:unit_results}
        
\end{figure}

\section{Conclusion}

We introduced and studied a minimal model to map out how and why heterogeneity in robotic teams affects performance in perimeter defense applications.


On the one hand, we showed that a heterogeneous team achieves better performance when full information of the oncoming attacks is available to the defenders. Moreover, we uncovered a seemingly universal behavior, where the ratio of optimal defender speeds is nearly constant for a range of problem parameters.


On the other hand, we proved that heterogeneity is detrimental to the system's performance in the converse case where minimal attack information is available. These results suggest that heterogeneity is potentially a non-robust property, since less system information dramatically decreases its usefulness.


Future directions involve quantifying and studying the use of heterogeneity when intermediate levels of information are available to the defenders. This would explore the existence of a phase transition where heterogeneity changes from decreasing to improving system performance. Possible scenarios include varying the horizon length of incoming attacks between the cases of $1$ and $\infty$ considered in the paper. Another scenario augments the unit time horizon with the knowledge of the number of remaining attacks. In particular, we conjecture that even in this case defenders should always capture attacks if possible and that heterogeneity remains detrimental. {Lastly, we wish to perform numerical simulations for a larger number of defenders.}

%

%
\section*{Appendix}

\subsection{Proofs}
\subsubsection{Proof of Theorem~1}
In this proof, for readability we denote $i^* := i^*(\bj)$ and $j^* := j^*(\bj)$; these values depend on $\bj$ (and are important to keep track of to reconstruct the optimal defense plan $\boldsymbol{\ell}$).

Alg.~1 depends on the function $f(\bj): \{0, 1, \dots, n\}^m \to \bbN$, which denotes the following: suppose that defender $i$ (with speed $v_i$) is required to thwart attack $j_i$ \emph{and then no others after that} (but defender $i$ can thwart attacks arriving before $t_{j_i}$, and if $j_i = 0$, then defender $i$ is not allowed to thwart any attack); $f(\bj)$ is the minimum number of defenders that can be let through under these constraints. Then the following hold:
\begin{itemize}
    \item $f(0,\dots,0) = n$ (the base case from which we recursively compute $f$);
    \item $\opt(\bv, \{(z_j,t_j)\}_{j=1}^n) = \min_{\bj} f(\bj)$ (this allows us to extract the correct value by keeping track of this minimum).
\end{itemize}
We then want to recursively compute $f(\bj)$ for all $\bj \in \{0, 1, \dots, n\}^m$. This can be done by considering 
\begin{align}
    \cS = \argmax_{i \in [m]} \, j_i
\end{align}
i.e. the set of defenders that thwart the latest attack in $\bj$. We then (arbitrarily) select a defender $i^* \in \cS$ to consider.  

We can then ask: suppose $j'$ is the last attack that defender $i^*$ thwarts before thwarting $j_{i^*}$. Then $j' < j_{i^*}$ and $\bone_{i^*}(j', j_{i^*}) = 1$, since otherwise $i^*$ cannot defend both attacks. The best the defenders can do in this case is to thwart $f(\bj_{-i^*}(j'))$ attacks, then have defender $i^*$ thwart $j_{i^*}$: if $|\cS| = 1$ ($i^*$ is the unique defender required to thwart $j_{i^*}$) then this is $1$ more attack thwarted in total; otherwise it's redundant. Thus, minimizing over all possible $j'$, we have
\begin{align}
    f(\bj) = \min_{j'} \big\{ f(\bj_{-i^*}(j')) : j' < j_{i^*} \text{ and } \bone_{i^*}(j', j_{i^*}) = 1  \big\} + \bone\{|\cS|=1\}
\end{align}
(where $\bone\{|\cS|=1\}$ is the indicator function for $|\cS| = 1$). We also let $j^*$ be (any) $j'$ which minimizes this, which will be important for reconstructing the optimal defense plan $\boldsymbol{\ell}$.

By iterating over all $\bj$ in lexicographic order, we can compute $f(\bj)$ using the above recursion (and starting from $f(0,\dots,0) = n$). We keep track of the minimum value (and the $\bj^{\min}$ which minimizes $f(\bj)$) and output this as $\opt(\bv, \{(z_j,t_j)\}_{j=1}^n)$.

To reconstruct the defense plan, we start with $\bj^{\min}$ and read backward: we know that the optimal defense plan (which is the optimal defense plan for $\bj^{\min}$) starts by defending according to $\bj^{\min}_{i^*}(j^*)$ and then having $i^*$ defend $j^{\min}_{i^*}$ at the end; we then recurse to $\bj^{\min}_{i^*}(j^*)$ until we arrive at $\bj = (0,\dots,0)$. \qed

\subsubsection{Proof of Proposition~1}

To begin, we want to give a rough justification for assuming condition (i) in Proposition 1. First, we conjecture that it is always optimal to thwart a reachable attack; furthermore, we prove that under certain conditions (specifically, if the number of future attacks is sufficiently small relative to a parameter dependent on the defender speeds) then thwarting it is always optimal.

Our setup is: the defenders are currently at a distance $a$ and an attack comes in at $x$ which can be reached by at least one of the defenders, after which $N$ further (uniformly and randomly distributed) attacks will follow. We have two policies: $f'$, which declines to capture in this case, and $f^*$, which follows the conjectured optimal strategy of always capturing, and always maximizing the distance between the defenders conditional on whether or not a capture is made. 

We let $J_{f'}(a; x; N)$ be the expected reward of following policy $f'$ from the given conditions (current separation $a$, attack at $x$, $N$ attacks after that), and $J_{f^*}(a; x; N)$ be the policy of following $f^*$ (including the initial capture which $f'$ does not do). Then if we can show that $J_{f'}(a; x; N) \leq J_{f^*}(a; x; N)$ it proves that $f^*$ outperforms all policies which do not make the capture at the current time. Since our result is such that if it works for this value of $N$, it also works for smaller values of $N$, it means that the optimal policy must always capture if possible and thus, by Proposition 1 in the main work, the policy $f^*$ is therefore optimal under the given conditions.

We also assume that the current separation between the defenders (before deciding whether to thwart the current reachable attack) is $a \geq 2v_2$; this assumption is justified because by Proposition 1 in the main work, we know that all things equal we want to maximize separation ($J$ is monotonic in $a$), and $2v_2$ separation can always be maintained even when capturing since the non-capturing defender can always maximize its distance to the current attack; even if the initial position of the defenders does not satisfy this, after $2$ attacks it can always be achieved (even with captures).


\begin{proposition}\label{prop::wproof}
Let $w = \min(1 - 2v_1, 1-v_1-3v_2, 2v_2, v_1-v_2)$, and let the current separation between the defenders be $a \geq 2v_2$; then if there are $N \leq \frac{1}{w}$ attacks left after the current attack, thwarting the current attack (if possible) is always optimal.
\end{proposition}

\begin{proof}
We assume that $2v_1 \leq 1$ and $v_1 + 3v_2 \leq 1$ (otherwise by Theorem 2 from the main work it is possible to thwart all attacks no matter what, which by definition means the optimal policy thwarts the current attack).

Let $s_{f'}(i)$ and $s_{f^*}(i)$ be the expected sizes of the union of the reachable set at step $i$ after the current step, under policies $f'$ and $f^*$ respectively, where $f'$ does not thwart the attack and $f^*$ is the conjectured optimal policy. Then $s_{f'}(i) \leq \min(1, 2(v_1+v_2))$ and $s_{f^*}(i) \geq \max(2v_1, v_1 + 3v_2)$. The first inequality is a generic upper bound on the size of the reachable set union for defenders of speeds $v_1, v_2$; the second holds because if no capture is made, the defenders can achieve at least $v_1 + v_2$ separation, and if a capture is made, they can achieve at least $2v_2$ separation (the non-capturing defender maximizes separation), in which case the union of the reachable sets satisfies the given bound.

Thus, the probability that the attack at step $i$ is reachable (with the current attack considered step $0$) is bounded by the above, and hence
\begin{align}
    J_{f'}(a; x; N) &\leq \min(1, 2(v_1+v_2))N
    \\ \text{and }~~ J_{f^*}(a; x; N) &\geq 1 + \max(2v_1, v_1 + 3v_2)N
\end{align}
Subtracting using these bounds, we get that
\begin{align}
    J_{f^*}(a; x; N) - J_{f'}(a; x; N) &\geq 1 \!-\! \min(1 - 2v_1, 1-v_1-3v_2, 2v_2, v_1-v_2)N
    \\ &= 1 - wN
    \\&\geq 0 ~~\text{ if } N \leq \frac{1}{w} \, .
\end{align}
The extra ``$+1$'' comes from the fact that $f^*$ makes a capture at the current step, while $f'$ does not. Thus, we are done.
\qed
\end{proof}

The value $w$ is maximized when $v_1 = 3/8$ and $v_2 = 1/8$, yielding $w = 1/4$. We also show $w$ as a function of $v_1$ and $v_2$ in \Cref{fig::wplot}, with the regions defined by Theorem~2 in the main text removed, since these guarantee perfect defense and hence that all attackers can and should be captured.

\begin{figure}
    \centering
    \includegraphics[width=0.7\linewidth]{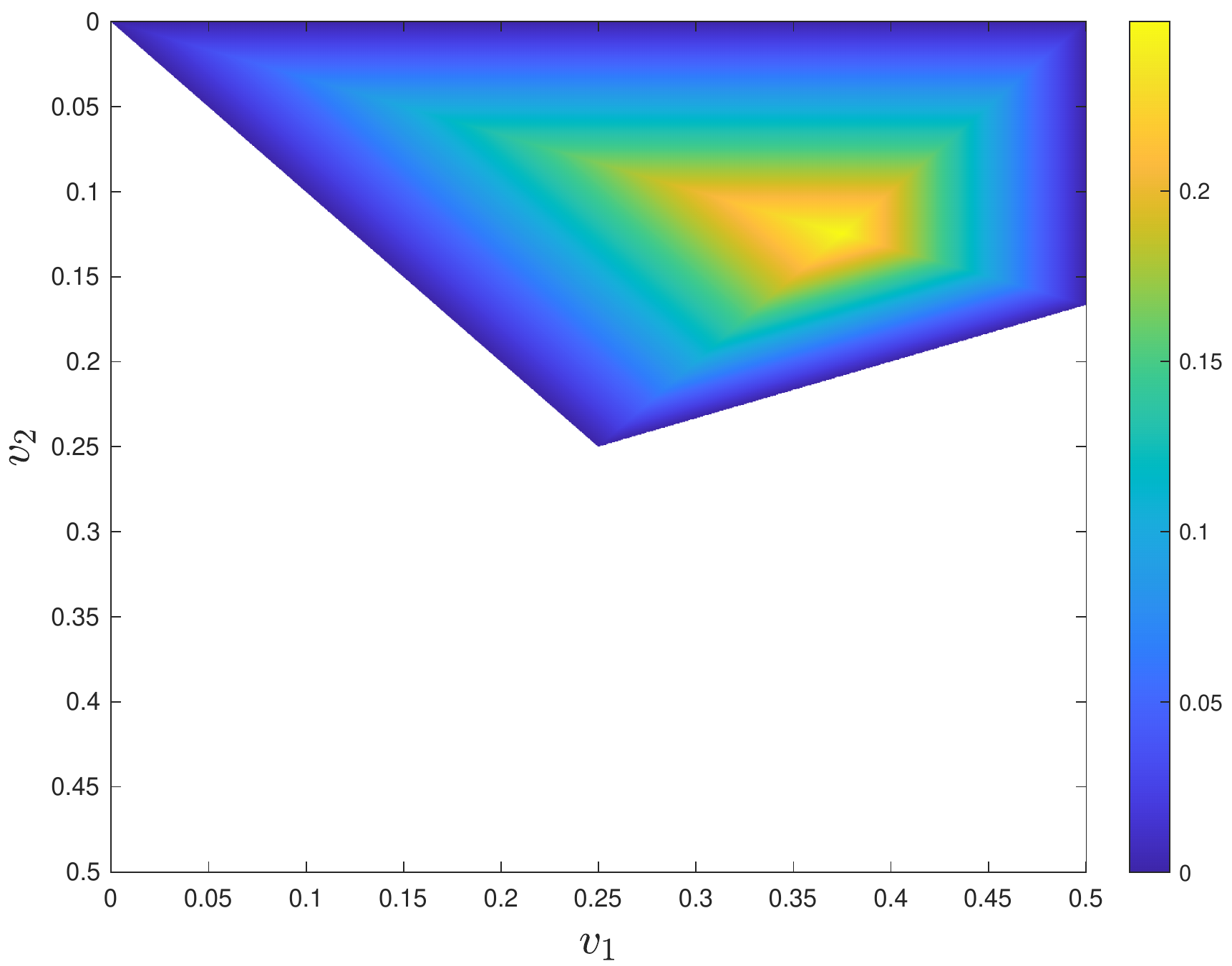}
    \caption{The quantity $w$ as a function of $v_1$ and $v_2$, with the regions defined by Theorem~2 in the main text removed.}
    \label{fig::wplot}
\end{figure}

Note also that the above holds as well if the number of future attacks is random (unknown to the defenders) with $\bbE[N]$ replacing $N$.

We note that the bound above is not the best achievable by this method; we can sharpen it by noting that under $f^*$, when a capture is \emph{not} made, the union of the reachable sets on the next step is $\min(1, 2(v_1 + v_2))$ (maximizing separation without needing to capture allows the defenders to have non-intersecting reachable sets).

We next proceed with the proof of Proposition~1 in the main text.
\begin{proof}[of Proposition~1]
The proof is by induction on the number of attacks $N$. We show $(ii)$ together with the following additional properties:
\begin{enumerate}[$(i)$]   \setcounter{enumi}{2}
    \item For any $N$, $J(a; N)$ is an increasing function of $a$;
\end{enumerate}
We start with the base case $N=0$. Note that given an initial defender distance $a$, the probability that the first attack can be defended is just the size of the union of the interval $[-{v_1}, {v_1}]$ (where the faster defender can reach) with the interval $[a-{v_2}, a+{v_2}]$ (where the slower defender can reach). We denote the size of this union as
\begin{align*}
    s(a) := s(a, {v_1}, {v_2}) := \begin{cases} 2{v_1} &\text{for } a \leq {v_1} - {v_2} \\ \min({v_1} + {v_2} + a, 2({v_1} + {v_2}), 1) &\text{for } a > {v_1} - {v_2}  \end{cases}
\end{align*}

Clearly, $J(a;0) \leq s(a)$, with equality if $(i)$ holds. Since $s(a)$ is a monotonically increasing function for $a \in [0,1/2]$, $(ii)$ holds. Clearly, $J(a;0) = s(a)$, so $(iii)$ holds as well.

Next, we proceed with the induction step and show $(ii)$. We make use of the following recursive formula:
\begin{equation}\label{eq::recursive-J}
J(a; N+1) = s(a) + \mathbb{E}_x [J(f(a,x); N)],   
\end{equation}
which holds since attack $N+1$ is defended with probability $s(a)$, by $(i)$, and by independence of the distributions of the attacks. Since $J(\cdot; N)$ is an increasing function, by $(iii)$, this expression is maximized by an $f$ maximizing $f(a,x)$, subject to the constraint in $(i)$. \eqref{eq::recursive-J} also shows $(iii)$, as $J(a; N+1)$ is the sum of two increasing functions. This concludes the induction step and therefore the proof.
\qed
\end{proof}

\subsubsection{Proof of Theorem~2}
We are using the parametrization of the perimeter as a unit-circumference circle from section 4.1 and allow the parametrization to wrap around the perimeter, i.e., e.g, $-1/4$ is equivalent to $3/4$.

\emph{Explicit attack sequence:} We first show that there exists a sequence of attacks that breaches if $v_1 < 1/2$ and $v_1 + 3v_2 < 1$. We devise an explicit attack pattern consisting of $6$ attacks which will successfully beat even very strong defenders. The key observations in designing this attack pattern are the following: 
\begin{enumerate}
    \item If $v_1 < 1/2$ (which we assume because otherwise all attacks can be defended) by setting $z_1 = -1/4$ and $z_2 = 1/4$ we can force one defender to be at $-1/4$ at time $1$ and the other to be at $1/4$ at time $2$.
    \item An attack pattern $(z_1, z_2, \dots, z_n)$ breaches if and only if the reverse pattern $(z_n, z_{n-1}, \dots, z_1)$ breaches. This is clear when we note that any valid movement of defenders is valid in reverse as well.
    \item An attack pattern $(z_1, z_2, \dots, z_n)$ breaches if and only if $(-z_1, -z_2, \dots, -z_n)$ also breaches.
    \item If we have a pattern $(-1/4, 1/4, x, y)$ which can only be defended by having the slower defender take the first attack, we can make it breach by expanding it to $(-y, -x, -1/4, 1/4, x, y)$. This is because to defend against the sequence starting at $-1/4$, we need the slower defender to take the attack at $-1/4$; but then to defend against $( -y, -x, -1/4, 1/4)$, by symmetry we need the slower defender to take the attack at $1/4$ (which corresponds to $-1/4$ in the original sequence). 
\end{enumerate}
Now we just need to find $x, y$ (there's no reason a priori to expect that $2$ additional attacks is the right number to use, it just happens to be what works) such that the sequence $(-1/4, 1/4, x, y)$ forces the slower defender to defend the second attack.

\paragraph{Finding $x, y$:} We assume that $v_1 < 1/2$ (otherwise it is trivial to defend against all attacks).

Suppose we have $-1/4, 1/4$ as the first two attacks (in that order) and the faster defender takes $-1/4$. Since $-1/4, 1/4$ are $1/2$ apart, it means the slower defender must defend $1/4$. We put the next attack at $x = 1/4 - v_2 - \varepsilon$; it is therefore out of reach of the slower defender and must be defended by the faster defender. We then put the last attack at $y = 1/4 - v_2 - v_1 - 2\varepsilon$, out of reach of the faster defender. It is out of reach of the slower defender too if $1/4 + 2 v_2$ is not sufficient to reach it (since the slower defender was at $1/4$ two time steps ago). Considering the negative-direction arc from $1/4$ to $1/4 - v_2 - v_1 - 2\varepsilon$, which has length $v_1 + v_2 + 2\varepsilon$, and the positive-direction arc from $1/4$ to $1/4 + 2v_2$, which has length $2v_2$, they fail to meet if their sum is less than $1$, i.e. if
\begin{align}
    v_1 + v_2 + 2\varepsilon + 2v_2 = v_1 + 3v_2 + 2\varepsilon < 1
\end{align}
Since we can choose $\varepsilon > 0$ as small as we like, this sequence works if $v_1 + 3v_2 < 1$ (and $v_1 < 1/2$). Thus for any pair of defenders satisfying this condition, for sufficiently small $\varepsilon > 0$ (in particular $\varepsilon < (1 - (v_1 + 3v_2))/2$ and $\varepsilon < (1/2 - v_1)/2$), the attack pattern
\begin{align}
\begin{split}
    \Big((v_1 + v_2 + 2\varepsilon) - 1/4, &(v_2 + \varepsilon) - 1/4, -1/4, 1/4, \\
    &1/4- (v_2 + \varepsilon), 1/4 - (v_1 + v_2 + 2\varepsilon)\Big)
    \end{split}
\end{align}
cannot be defended.

\paragraph{Perfect defense:} We conclude by showing the reverse direction, namely that any team of defenders such that $v_1 + 3v_2 \geq 1$ can always defend against any sequence of attacks. To show this, it is clearly sufficient (and necessary) to show it for the case of $v_1 + 3v_2 = 1$. To do this, we show the stronger result that this team can defend all attacks even with a $1$-step visibility horizon (i.e. the defenders only need to know the location of the next attack).

We say that the defenders have \emph{full coverage} if the union of their reachable sets is the entire perimeter. This means that, for the next time step at least, the defense cannot be breached. If they can guarantee full coverage at every step then they can defend against any sequence of attacks. Note that since $v_2 \leq v_1$, we have
\begin{align}
    1 =  v_1 + 3v_2 = (v_1 + v_2) + 2v_2
\end{align}
Since $v_1 + v_2 \geq 2v_2$ and $(v_1 + v_2) + 2v_2 = 1$, we can conclude that $v_1 + v_2 \geq 1/2$ (with equality iff $v_2 = v_1 = 1/4$), and hence the two defenders have reachable sets large enough to get full coverage.

Let $s(t)$ be the distance between the defenders at time $t$. Then they have full coverage at step $t$ if and only if \begin{align}
    s(t) \geq 1/2 - (v_1 + v_2 - 1/2) = 1 - (v_1 + v_2) = 2v_2
\end{align}
Our goal is now to show that if $s(t) \geq 2v_2$ then no matter where $x(t+1)$ (the next attack) happens, it can be defended in such a way that $s(t+1) \geq 2v_2$; then by induction, as long as $s(0) \geq 2v_2$ (which we guarantee by starting them off at opposite poles, i.e. $s(0) = 1/2$).

If the attack is reachable by the slower defender ($v_2$ to either side) it is easy to see that the slower defender can defend against it and the faster defender can imitate the movement (same distance and direction) to make $s(t+1) = s(t) \geq 2v_2$. If the attack is outside of the reachable set of the slower defender, then the faster defender must take it; but since it is outside the $2v_2$-wide reachable set of the slower defender, one of the two endpoints of that reachable set must be at least $2v_2$ away from attack location. The slower defender then goes to that endpoint, thus preserving $s(t+1) \geq 2v_2$.

Thus, we know that defenders satisfying $v_1 + 3v_2 \geq 1$ can defend against any sequence of attacks, even with a one-step horizon.
\qed

\subsection{Homogeneous defenders}

For a team of homogeneous defenders with speed $v$ against a sequence of attacks $\{(z_j, t_j)\}_{j=1}^n$, we show the following matching-based algorithm for determining the minimum number $m$ of such defenders are needed to thwart \emph{all} the attacks. The homogeneity of the defenders is key to this algorithm, which runs in time $O(\sqrt{n} |E|)$ where $|E|$ is the number of edges in the DAG $G$; this is in turn proportional to $n^2$ in the worst case and hence the overall run time is $O(n^{5/2})$.

First, we build a directed acyclic graph (DAG) $G$ on $n$ nodes $u_1, \dots, u_n$, each node representing an attack; we assume that they are sorted by time $t_1 \leq t_2 \leq \dots t_j$ (if not, we sort them in $O(n \log n)$ time). We put a directed edge $u_{j} \to u_{j'}$ (where $j' > j$) if and only if $\dist(z_{j}, z_{j'}) \leq v (t_{j'} - t_{j})$, i.e. if attack $j'$ can be reached by a defender which thwarts attack $j$.

Note that a (directed) path in $G$ corresponds to a sequence of attacks that can be defended by a single defender. Thus, the goal is to decompose $G$ into $m$ directed paths which cover all vertices.

This path decomposition is done via a (well known) reduction to maximum bipartite matching: we split each node $u_j$ into $u^{in}_j$ and $u^{out}_j$, where a directed edge $u^{in}_j \to u^{out}_j$ exists, and all edges $u_{j} \to u_{j'}$  are replaced with edges $u^{out}_{j} \to u^{in}_{j'}$. This does not change the minimum $m$.

Taking $\{u^{out}_j\}$ and $\{u^{in}_j\}$ as the two parts of a bipartite graph (and ignoring the $u^{in}_j \to u^{out}_j$ edges), any matching on this graph can be used to reconstruct a set of nonoverlapping paths that cover all the vertices of $G$ by putting back the $u^{in}_j \to u^{out}_j$ edges; since each $u^{in}_j$ has exactly one edge coming out of it (to $u^{out}_j$) and each $u^{out}_j$ has at most one edge coming out of it (if it is matched it has one, otherwise not) one can start at any unmatched $u^{in}_j$ and uniquely walk forward until reaching some unmatched $u^{out}_{j'}$, producing one directed path. 

Since there is one such path starting at each unmatched $u^{in}_j$ (and ending at each unmatched $u^{out}_{j'}$), a matching of size $k$ produces a set of $m = n - k$ directed paths covering the $G$ (and conversely any set of $m$ directed paths can be used to find a size-$(n-m)$ sized matching). Thus, finding the maximum $k$ gives the minimum numbers of defenders $m$.


Thus, we can use any well-known maximum bipartite matching algorithm (e.g. Hopcroft-Karp) to compute the number of defenders needed.

\begin{remark}
Note that a complete matching (which would imply that $0$ defenders could thwart all the attacks!) is by definition impossible because $G$ is acyclic.
\end{remark}

\noindent The algorithm is summarized in Alg.~\ref{alg:one} for ease of reference.
\begin{algorithm}
\caption{Homogeneous case.}\label{alg:one}
\KwData{Sequence of attacks $\{(z_j, t_j)\}_{j=1}^n$}
\KwResult{Minimum number of defenders $m$ required to defend all attacks.}
 Form DAG $G$ with nodes $u_1, \dots, u_n$ and edge $u_{j} \to u_{j'}$ iff $\dist(x_{i_1}, x_{i_2}) \leq s(t_{i_2} - t_{i_1})$\;
Split each node $u_i$ into $u^{in}_i$ and $u^{out}_i$. Add edges $u^{in}_i \to u^{out}_i$, and replace all edges $u_{i_1} \to u_{i_2}$ with edges $u^{out}_{i_1} \to u^{in}_{i_2}$\;
 Use the Hopcroft-Karp algorithm to find a maximum-bipartite matching on $G$ of size $k$\;
 $m \gets n - k$
\end{algorithm}

\subsubsection{Complexity}

There are two parts to this problem, each of which could probably be sped up with a little thought:
\begin{enumerate}
\item Building DAG $G$ and/or its bipartite counterpart.
\item Computing the maximum matching.
\end{enumerate}
Ignoring $d$, which we assume to be small, (1) takes $O(n^2)$ time, and (2) with Hopcroft-Karp takes $O(\sqrt{n} |E|)$ where $E$ is the number of edges, which in the worst case means $O(n^{5/2})$. 

\subsubsection{Simulation Results}

The goal of the simulations is to explore the relationship between speed and number of defenders (if we halve the speed of the defenders, how many more do we need to keep the same level of effectiveness?)

Our experiments follow the pattern: 
\begin{enumerate}
    \item Generate attacks $\{(z_j, t_j)\}_{j=1}^n$ randomly, either with fixed attack times $t_j = j$ or uniformly random attack times in $[0,t_{\max}]$.
    \item Compute $M(v, \{(z_j, t_j)\}_{j=1}^n) $ for $v \in (v_{\min}, v_{\max}]$ or $\bv \in (v_{\min}, v_{\max}]^m$, at $g$ intervals.
    \item Repeat the above for $T$ trials and average the resulting values for each value of $v$ or $\bv$.
\end{enumerate} 

We compute $M(v, \{(z_j, t_j)\}_{j=1}^n)$ for all $\bv$ whose $v_i$ are at $g$ evenly-spaced locations in the range $(v_{\min}, v_{\max}]$ of speeds. We refer to $g$ as the number of \emph{grains}.


The full list of parameters is given in \Cref{tab:parameters}.

\begin{table}[!t]
	\centering
	\caption{Parameters of the experiments}
	\label{tab:parameters}
	\begin{tabular}{@{}ll@{}}
		\toprule
		Symbol         & Description                    \\ \midrule
		$m$          & Number of defenders ($m=2$ unless specified otherwise)              \\
		$n$          & Number of attacks               \\
		$T$ & Number of trials \\
		$t_{\max}$ & Size of attack window (not needed for heterogeneous setting (ii)) \\
		$(v_{\min}, v_{\max}]$ & Range of defender speeds (inclusive of $v_{\max}$ but not $v_{\min}$) \\
		$g$ & Number of speed values measured (grains) within $(v_{\min}, v_{\max}]$ \\
		\bottomrule
	\end{tabular}
\end{table}


For ease of notation, we denote
\begin{align}
    \bbE[M(v)] := \bbE[M(v, \{(z_j, t_j)\}_{j=1}^n)]
\end{align}
The expectation is over the attacks, i.e. random $(z_j, t_j)$. All experiments are for uniformly-random attack times on a circular ($1$-dimensional) perimeter of circumference $1$.

Our simulations suggest a strong relationship between the speed $v$ of the homogeneous defenders and the expected number $\bbE[M(v)]$ of defenders required to thwart all attacks: plotting these in a log-log plot reveals an almost linear relationship between $\log v$ and $\log \bbE[M(v)-1]$, as shown in \Cref{fig:speed-vs-num}.\footnote{We subtract $1$ because $1$ defender is always required, so this measures the number of additional defenders above the minimum.}

\begin{figure}[h!]
     \centering
     \begin{subfigure}[b]{0.45\textwidth}
         \centering
         \includegraphics[width=\textwidth]{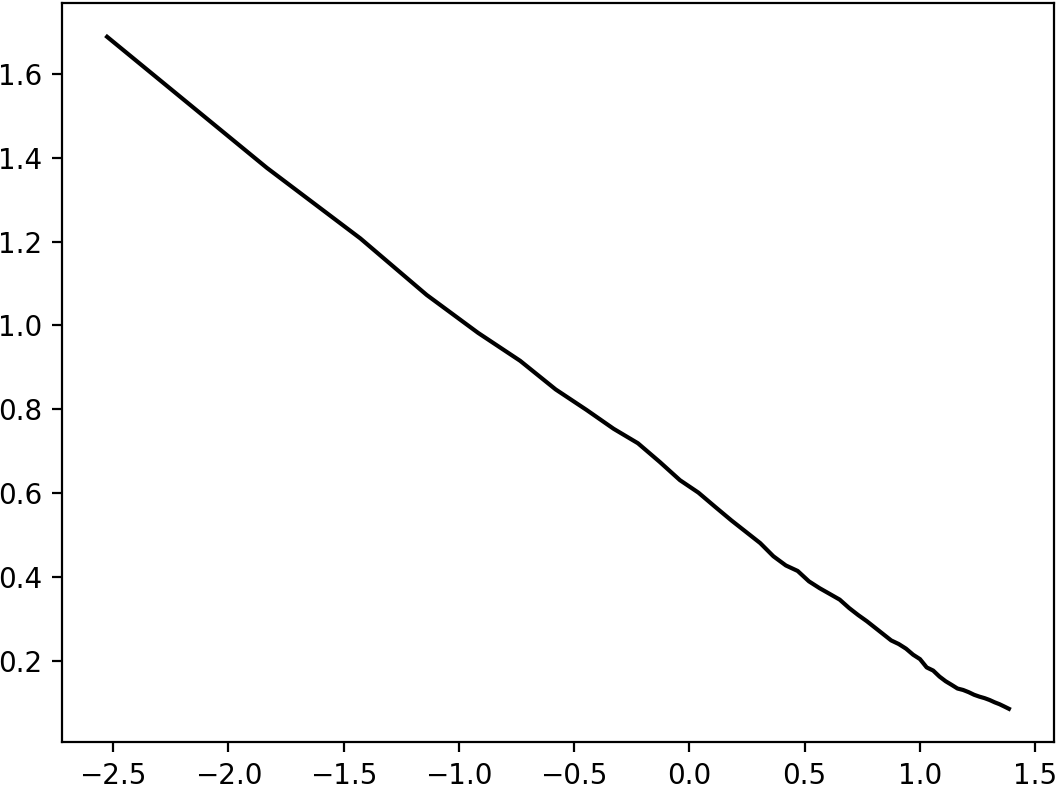}
         \caption{$t_{\max} = 15, n = 25, T = 1000$}
         \label{fig:speed-vs-num-d01_h15_n30_tr100}
     \end{subfigure}
     \hfill
     \begin{subfigure}[b]{0.45\textwidth}
         \centering
         \includegraphics[width=\textwidth]{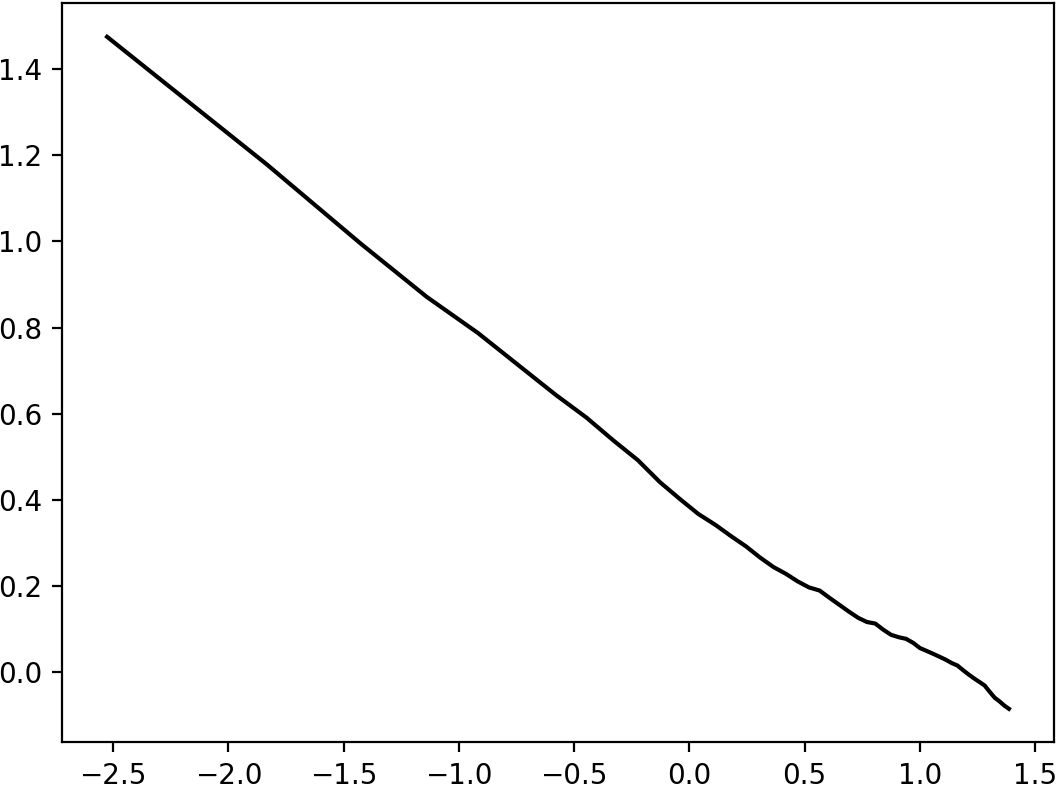}
         \caption{$t_{\max} = 25, n = 25, T = 1000$}
         \label{fig:speed-vs-num-d01_h25_n30_tr100.png}
     \end{subfigure}
        \caption{Log-log plots of $v$ versus (empirical) $\bbE[M(v)-1]$, plotted for a range of $50$ speeds $v$ from $0$ (not inclusive) to $4$ (inclusive), i.e. $v = 0.08, 0.16, \dots, 4.00$, for two different values of $t_{\max}$. Each plot done for $T = 1000$ trials of $n = 25$ random attacks (uniformly-random attack times).}
        \label{fig:speed-vs-num}
\end{figure}
        
This suggests that given a particular distribution of attacks, there is some value $\alpha > 0$ such that
\begin{align}
    \bbE[M(cv)-1] \approx c^{\alpha} \, \bbE[M(v)-1]
\end{align}
The slope of the line\footnote{Since it is not a 100\% perfect line (only extremely close) the slope is estimated via linear regression.} in the log-log plots in \Cref{fig:speed-vs-num} is $-\alpha$. What $\alpha$ measures, roughly, is the balance of the speed-number tradeoff for the defenders: if we double the speed of the defenders, how many fewer defenders (on average) will we need to stop all the attacks? The nearly linear nature of \Cref{fig:speed-vs-num} indicate that, for a given (uniform in time and location) distribution of attacks, this tradeoff is roughly constant over a wide range of speeds. When $\alpha$ is small, it indicates that increasing the number of defenders is comparatively more important than increasing the speed (by the same proportion); in particular, if $\alpha = 1$ then doubling the speed and doubling the number are equivalent.

We then examine what this $\alpha$ is given different distributions of attacks: for a given value of $t_{\max}$, we vary the number $n$ of attacks and see how $\alpha$ changes. The results are shown in \Cref{fig:alpha-params}, with the $y$-axis being the slope of the log-log plot of the given parameters ($-\alpha$, as mentioned above). Note that $\alpha < 1$ in all measured cases, meaning that increasing the number of defenders by some proportion $c$ is always more effective than increasing the speed of the defenders by a factor of $c$. When there are few attacks ($n = 5$), $\alpha \approx 0.8$, but as $n$ increases it drops to about $0.4$ before leveling out, indicating a substantial decrease in the efficacy of increasing speed as compared to increasing the number of defenders. This general pattern holds true for both plotted cases ($t_{\max} = 15$ and $t_{\max} = 25$) but the drop-off is sharper when $t_{\max} = 15$.

\begin{figure}[h!]
     \centering
     \begin{subfigure}[b]{0.48\textwidth}
         \centering
         \includegraphics[width=\textwidth]{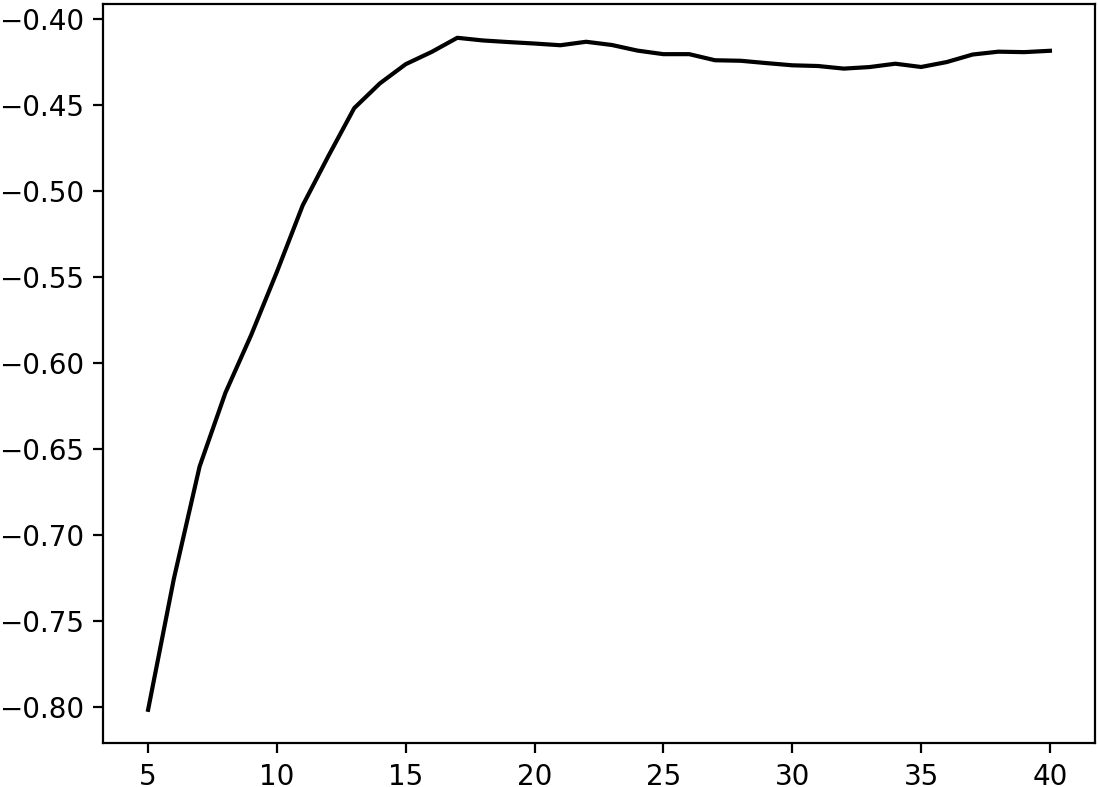}
         \caption{$t_{\max} = 15$}
         \label{fig:alpha-params-d01_h15_s4-50_tr1000.png}
     \end{subfigure}
     \hfill
     \begin{subfigure}[b]{0.48\textwidth}
         \centering
         \includegraphics[width=\textwidth]{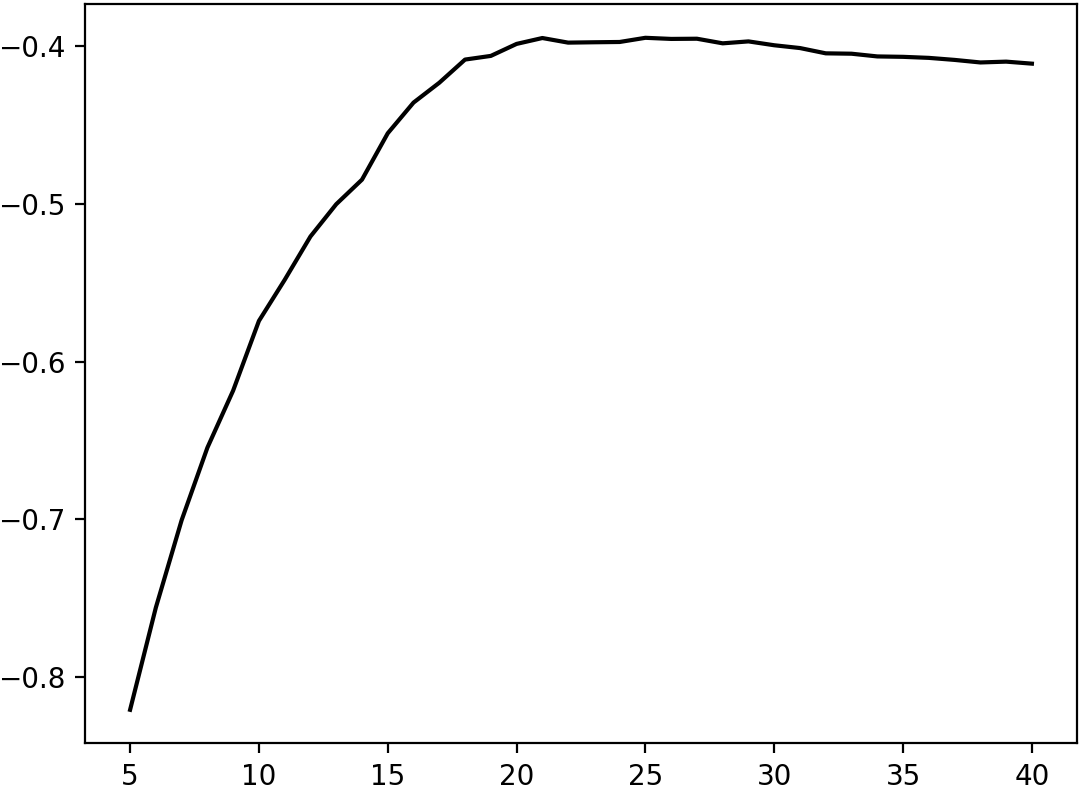}
         \caption{$t_{\max} = 25$}
         \label{fig:alpha-params-d01_h25_s4-50_tr1000.png}
     \end{subfigure}
        \caption{Relationship of $-\alpha$ ($y$-axis) to number of attacks $n$ ($x$-axis) from $n = 5$ to $n = 40$, uniform attack times, for two different values of $t_{\max}$. For each value of $n$ and each of $T = 1000$ random attack sequences, the value $M(v)$ was computed for $50$ linearly-spaced values of $v$ from $0$ to $4$ (not inclusive of $0$), i.e. $v = 0.08, 0.16, \dots, 4.00$, and the best-fit slope of the log-log plot extracted.}
        \label{fig:alpha-params}
\end{figure}

\subsection{The Monotonicity Acceleration}

We now describe in greater detail how we evaluate $\opt(\bv)$ in an unusual order in order to take advantage of the monotonicity of the function. While we have not made any rigorous attempt to optimize the ordering, the order we use was effective enough to reduce the computation of the simulations described in the main work (see Figures 3 and 5, left-hand column for uniformly-random attack times and right-hand for fixed attack times) by an average of between $93\%$ and $99\%$ depending on the parameters; the savings are much more pronounced for the fixed attack times case, as the $\opt$ function is generally flatter (and hence produces more cases where the upper bound matches the lower bound and the DP computation can be skipped).

First, it is noted that permuting $\bv$ does not affect $\opt(\bv)$; hence whenever $\opt(v_1, v_2)$ is computed, we get the same value for $\opt(v_2, v_1)$ (recall that in these simulations we drop the WLOG assumption that $v_1 \geq v_2$). For simplicity we will talk about the \emph{indices} of the speed values we compute for: the goal is to compute $\opt(\bv)$ for all $\bv$ with entries from a range of $g$ values, which we will denote as a vector $\bw = (w_1, \dots, w_g)$ where $w_1 < w_2 < \dots < w_g$ without loss of generality. Then we define a function $f: [g]^m \to \bbR$ where
\begin{align}
    f(j_1, \dots, j_m) = \opt(w_{j_1}, \dots, w_{j_m})
\end{align}
This just means that the $i$th defender has speed $v_i = w_{j_i}$. Since $\bw$ is monotonically increasing (in its entries) and $\opt$ is monotonically decreasing, $f$ is monotonically decreasing as well. For this section, as in the simulations, $m = 2$ so we can visualize $[g]^m$ as a grid (as in Figure 5 in the main paper).

The basic idea is to compute in order of decreasing powers of $2$, which we refer to as \emph{levels}: for instance, when $g = 256$ we compute first for multiples of $256$, then multiples of $128$, then $64$ and so forth. In this way, at each level we can use the results from the level above to check the monotonicity condition and potentially save significant computation.

The algorithm then does the following:
\begin{enumerate}
    \item Compute $f(j_1, j_2)$ for the four corners of the grid.
    \item Compute $f(j_1, j_2)$ on the edges and main diagonal of the grid, in decreasing powers of $2$. This establishes a border so that every subsequent step has an upper bound and a lower bound to compare.
    \item Compute $f(j_1, j_2)$ by decreasing levels (at each level, doing it in lexicographic increasing order as normal, skipping any entries already computed previously). 
\end{enumerate}
We also for simplicity only considered the upper and lower bounds generated by previous entries on the same row or column of the grid (two upper bounds, one for the row and one for the column, to two lower bounds, also one for the row and one for the column).

%


\end{document}